\documentclass[journal]{IEEEtran}
\usepackage{cite}
\usepackage{amssymb}

\ifCLASSINFOpdf
\usepackage[pdftex]{graphicx}
\else
\usepackage[dvips]{graphicx}
\fi
\usepackage[cmex10]{amsmath}
\usepackage{mdwmath}
\usepackage{amsmath,amssymb,amsfonts}
\newcommand{\xoic}{\boldsymbol{x}^{[c_{\mathrm{o}}+1,c_{\mathrm{i}}]}}
\newcommand{\yoic}{\boldsymbol{y}^{[c_{\mathrm{o}}+1,c_{\mathrm{i}}]}}
\newcommand{\poic}{\boldsymbol{p}^{[c_{\mathrm{o}}+1,c_{\mathrm{i}}]}}
\newcommand{\xoici}{\boldsymbol{x}^{[c_{\mathrm{o}}+1,c_{\mathrm{i}}-1]}}
\newcommand{\yoici}{\boldsymbol{y}^{[c_{\mathrm{o}}+1,c_{\mathrm{i}}-1]}}
\newcommand{\poici}{\boldsymbol{p}^{[c_{\mathrm{o}}+1,c_{\mathrm{i}}-1]}}
\usepackage{amsthm}
\newtheorem{theorem}{Theorem}   

\usepackage{algorithm}
\usepackage[noend]{algorithmic}

\usepackage{tabularx}
\usepackage{graphicx}
\usepackage{textcomp}
\usepackage{xcolor}
\usepackage{stfloats}
\usepackage{multirow}
\usepackage{booktabs}
\usepackage{array}
\usepackage{subfigure}
\usepackage{makecell}
\usepackage{array}
\usepackage{booktabs}
\usepackage{overpic}
\usepackage{enumitem}
\usepackage{commath}
\usepackage{mathtools}
\usepackage{gensymb}
\usepackage{extarrows}
\allowdisplaybreaks[4]
\usepackage[colorlinks, linkcolor=black, anchorcolor=black, urlcolor=black, citecolor=black]{hyperref}

\def\BibTeX{{\rm B\kern-.05em{\sc i\kern-.025em b}\kern-.08em
		T\kern-.1667em\lower.7ex\hbox{E}\kern-.125emX}}
\newcommand{\leoname}{M-LEO}
\newcommand{\saginname}{SAGINs}
\newcommand{\methodname}{DV-MOSS}
\newcommand{\stagename}{JUBPA}
\begin{document}
\bstctlcite{IEEEexample:BSTcontrol}
	
	\title{Visibility-aware Satellite Selection and Resource Allocation in Multi-Orbit LEO Networks}
\author{Yingzhuo Sun,\IEEEmembership{~Student Member, IEEE,}
Yulan Gao,\IEEEmembership{~Member, IEEE,}
Ming Xiao,\IEEEmembership{~Senior Member, IEEE,}
Zhu Han,\IEEEmembership{~Fellow, IEEE,} and 
Octavia A. Dobre,\IEEEmembership{~Fellow, IEEE}
\thanks{This paper has been presented in part at the IEEE International Conference on Communications (ICC), 2025.}
\thanks{Y. Sun, Y. Gao, and M. Xiao are with the Division of Information Science and Engineering, KTH Royal Institute of Technology, 100 44 Stockholm, Sweden
(e-mail: \{yingzhuo, yulang, mingx\}@kth.se).}
\thanks{Z. Han is with the Department of Electrical and Computer Engineering, University of Houston, Houston, TX 77004, USA (e-mail: hanzhu22@gmail.com).}
\thanks{O. A. Dobre is with the Faculty of Engineering and Applied Science, Memorial University, St. John’s, NL A1B 3X5, Canada (e-mail: odobre@mun.ca).}
}

\maketitle
\begin{abstract}
Multi-orbit low-earth-orbit (\leoname{}) satellites communication is envisioned as a key infrastructure to deliver truly global coverage, enabling future services from space–air–ground integrated networks.
However, the optimized design of \leoname{} which jointly addresses satellite selection, association control, and resource scheduling while accounting for dynamic visibility in multi-orbit constellations still remains open. Satellites moving along distinct orbital planes yield phase-shifted ground tracks and heterogeneous, time-varying coverage patterns that significantly complicate the optimization.
To bridge the gap, we propose a dynamic visibility-aware multi-orbit satellite selection framework which can determine the optimal serving satellites across orbital layers. The framework is built upon Markov approximation and matching game theory to jointly address the challenges of time-varying visible satellite sets, fluctuating link quality due to varying visibility, and dynamic inter-user interference.
Specifically, we formulate a combinatorial optimization problem that maximizes the sum-rate under per-satellite power budgets and multi-orbit constellation size limitations. 
The problem is NP-hard because it entails a mixed-integer non-convex structure, combining discrete user association (UA) decisions with continuous power allocation, and an inherently non-convex sum-rate maximization objective. We address it through a problem-specific Markov approximation. 
Moreover, we alternately solve UA/bandwidth allocation via a matching-game mechanism and power allocation via a Lagrangian dual convex program, which together form a block coordinate descent method tailored to this problem.
Simulation results show that the proposed algorithm converges to a suboptimal solution across all scenarios.
Extensive experiments against four state-of-the-art baselines further demonstrate that our algorithm achieves, on average, approximately 7.85\% higher sum rate than the best-performing baseline.
\end{abstract}
	
	
\section{Introduction}\label{introduction}

\IEEEPARstart{W}{ith} ever-growing demand for higher data rates and ubiquitous connectivity, the existing terrestrial networks are increasingly unable to satisfy the requirements of future communication services \cite{6G_1}, and leave users in remote, rural, and maritime areas underserved due to the lack of coverage \cite{LEO1}.
To overcome these limitations, satellite communications have emerged as a promising complement to terrestrial networks. 
In particular, multi-orbit low-earth-orbit (\leoname{}) satellite systems are envisioned as an infrastructure to provide seamless connectivity with truly global coverage, bridging underserved regions and enhancing the resilience of space-air-ground integrated networks (\saginname{}) \cite{LEO_potential}. 

\leoname{} systems are typically defined as satellites operating in orbits at altitudes between $500\mbox{ km}$ and $2{,}
000\mbox{ km}$. Their primary advantages include low launch costs, and low round-trip delay owing to the relatively short propagation distance. However, deploying large-scale \leoname{} constellations inevitably introduces significant network complexity due to the massive number of satellites and high velocity. 
In particular, the resulting highly dynamic constellation leads to a  time-varying set of visible satellites over any given region, thereby making adaptive satellite selection essential for ensuring reliable connectivity to user equipment (UEs).

Extensive research efforts have been devoted to addressing the challenges of satellite constellation design, user association (UA), and resource allocation (RA).
Among the early contributions in this area, \cite{how_many_sats} developed a three-dimensional constellation design algorithm to minimize the number of satellites while ensuring backhaul requirements.
Another corpus of literature has focused on UA and RA under fixed constellation settings.
For instance, \cite{RAwithoutSS3} formulated the UA problem between ground UEs and LEO satellites, but the study assumed only three fixed satellites and did not consider satellite selection.
Beyond fixed settings, several studies investigated satellite selection strategies.
To capture long-term selection dynamics, \cite{LongtermSSNoRA1} introduced a non-dominated sorting optimization algorithm for \leoname{} selection in global navigation satellite systems, considering signal strength, and residual visibility time.

Despite these advances, most existing studies assume fixed constellation, where the set of serving satellites is predetermined and static. 
These common assumptions may be reasonable for simplified scenarios with a small number of satellites, but they become unrealistic in large-scale \leoname{} systems. 
In practice, satellites in different orbital planes move continuously. This leads to  a dynamic and time-varying set of visible satellites over any specific region. 
This dynamic visibility introduces two critical challenges: 
\begin{enumerate}
    \item the number of visible satellites fluctuates significantly over time, ranging from dense situations with many candidate satellites to sparse conditions with only a few feasible connections;
    \item even if multiple satellites are simultaneously visible, their link qualities vary due to differences in elevation angles, residual visibility time, and interference conditions.
\end{enumerate}
These factors imply the number of satellites to be selected and the serving satellites must be optimized dynamically.

To address the challenges, we propose the dynamic visibility-aware multi-orbit satellite selection framework (\methodname{}), which integrates constellation-level satellite adaptation with fine-grained per-satellite selection across multiple orbital layers.
Particularly, we investigate a scenario in which a UE is simultaneously covered by multiple candidate satellites. Then the network controller must optimize the number of serving satellites, the satellite-UE associations, and the transmit powers. 
The objective is to maximize the sum-rate under per-satellite power budget constraints and overall constellation size limitations.
Notably, the number of visible \leoname{} satellites varies significantly across time slots.  This dynamic visibility makes the problem more challenging. 
The main innovations and contributions of this paper can be summarized as follows.
\begin{itemize}
\item We investigate the sum-rate maximization problem of the joint optimization of multi-orbit satellite selection, association control, and PA under constellation size constraints in an \leoname{} setting. 
Different from existing works assuming a fixed set of serving satellites, we explicitly model the temporal evolution of candidate satellite sets.  
Dynamic visibility together with constellation sizes changes the problem setting drastically. Since satellites in different orbital planes and with varying relative phases continuously enter and exit the cone-shaped visibility regions of ground UEs,  the feasible set of the problem is inherently time-varying. 
The formulated problem is NP-hard because it is a mixed-integer non-convex program: discrete UA variables are coupled with continuous power variables through a non-concave sum-rate objective and a non-concave minimum-rate constraint.
\item To solve the problem, we propose \methodname{} algorithms combining Markov approximation and  matching game theory.
Specifically, we design an ergodic Markov chain with tailored transition probabilities that guarantee convergence to its stationary distribution. 
Building on this reformulation, we develop a block coordinate descent (BCD) architecture, termed \stagename{}, for the joint optimization of UE association, bandwidth allocation (BA), and power allocation (PA).
In particular, matching game theory is employed to determine stable satellite-UE associations together with BA, whereas a Lagrangian dual-based method is adopted to optimize PA.

This integrated approach achieves sub-optimal performance with provable convergence of submodules and significantly lower complexity than heuristic alternatives. 
\item We conduct extensive experiments on the Walker Constellation Design Pattern\footnote{The Walker constellation has been widely adopted in practical deployments, including Starlink, OneWeb, and Kuiper systems. It is available at \url{https://innovationspace.ansys.com/courses/courses/intro-to-orbit-types/lessons/walker-constellation-lesson-9/}.} to validate the convergence properties of \methodname{} and to evaluate its performance.
The results demonstrate that \methodname{} consistently outperforms the best-performing baseline, achieving an average improvement of approximately 7.85\% in sum rates.
\end{itemize}

The remainder of the paper is organized as follows. Section~\ref{related work} discusses the related work. Section~\ref{System Model} describes the system model and problem formulation. The algorithm design and simulation results are presented in Section~\ref{Algorithm Design} and \ref{Simulation Results}, respectively. Finally, the conclusions are drawn in Section~\ref{Conclusions}.

\section{Related Work}\label{related work}
Research on \leoname{} satellite communications and RA in \saginname{} has gained significant attention in recent years.
Prior studies provide valuable insights into constellation design, UA, and RA, as well as satellite selection strategies. 
However, most of these works rely on simplified assumptions, such as fixed satellite constellations or the separate optimization of UA and RA, which limit applicability in large-scale, dynamic \leoname{} systems. 
Existing results can be roughly divided into four categories, as follows. 

{\em 1) Optimization with Fixed Constellations:}
A commonly adopted paradigm for studying UA and RA in satellite networks is to assume a fixed LEO satellite constellation. For instance, \cite{RAwithoutSS1, RAwithoutSS3, RAwithoutSS4,RAwithoutSS5,RAwithoutSS6} represent typical research based on this paradigm. 
Reference \cite{RAwithoutSS3} developed a 3D constellation coordinate model and aimed to maximize the minimum average throughput using an iterative algorithm that combines block coordinate descent with successive convex approximation. 
Similarly, \cite{RAwithoutSS4} introduced a two-way transmission model with unknown channel state information and designed a hierarchical multi-agent multi-armed bandit-based RA scheme. 
In \cite{RAwithoutSS5}, UA and RA were jointly optimized by formulating a problem that incorporates both base station backhaul and UE traffic demands, and is solved via binary relaxation and alternating optimization.
A deep reinforcement learning-based RA strategy was further proposed in \cite{RAwithoutSS6} to address low resource utilization caused by uneven traffic demands. 

{\em 2) Single-Satellite Selection Approaches:}
A few studies aim to optimize selecting a single serving LEO satellite for each UE. Specifically, \cite{SingleSS1} modeled the selection as a potential game, treating satellites and frequency bands as shared resources while UEs act as competing players. 
A random access algorithm within a software-defined networking framework was then developed to maximize the area sum-rate. 

{\em 3) Multi-Satellite or Multi-Orbit Selection Methods:}
Beyond single-satellite settings, the literature \cite{StaticMultipleSS1,StaticMultipleSS3,StaticMultipleSS4} investigated the selection of multiple LEO satellites in static systems. 
In \cite{StaticMultipleSS1}, the authors proposed a two-layer algorithm that combines difference-of-convex programming with coalition formation to jointly address beamforming and UA.
A comparative study in \cite{StaticMultipleSS3} analyzed three UA schemes using stochastic geometry as the evaluation framework.
More advanced optimization techniques were explored. Reference \cite{StaticMultipleSS4} applied alternating direction method of multipliers (ADMM), bisection, and interior-point methods to jointly solve satellite selection and access problems, targeting improvements in downlink sum-rate, uplink fairness, and energy efficiency.

{\em 4) Long-Term Multi-Satellite Selection:}
A few works investigated long-term strategies for selecting multiple LEO satellites. 
For instance, \cite{LongtermSSNoRA1,LongtermSSNoRA2,LongtermSSNoRA3,LongtermSSNoRA4,LongtermSSNoRA6} considered satellite selection without explicit UA or RA. 
In \cite{LongtermSSNoRA1}, an iterative optimization algorithm for the Global Navigation Satellite System systems was developed, incorporating signal strength, residual visibility time, and receiver sensitivity. 
The authors in \cite{LongtermSSNoRA2} addressed the coexistence between primary and secondary systems via the Lagrange relaxation and sub-gradient methods. 
Reference \cite{LongtermSSNoRA3} applied gated recurrent units to predict trajectories and design handover strategies for high-speed terminals. 
A dynamic greedy algorithm for handover based on UE grouping was proposed in \cite{LongtermSSNoRA4}. 
In \cite{LongtermSSNoRA6}, a conditional handover optimization strategy using service continuity graphs was designed to maximize a reward function and predict handover sequences while improving the stability. 

Meanwhile, references \cite{LongtermSSRA1,LongtermSSRA2,LongtermSSRA3,LongtermSSRA5,LongtermSSRA6} incorporated UA or RA alongside multi-satellite selection. 
In \cite{LongtermSSRA1}, a distributed massive multiple-input multiple-output-based network architecture was introduced, and an artificial intelligence-enabled joint PA and handover management scheme was proposed to enhance throughput and reduce handovers. 
Reference \cite{LongtermSSRA2} proposed a collaborative beamforming paradigm with an evolutionary multi-objective deep reinforcement learning algorithm, improving uplink performance for energy-limited terminals. 
In \cite{LongtermSSRA3}, the satellite-aircraft handover was modeled as a local cooperation game, where subcarrier and PA were jointly optimized via a parallel update strategy. 

Caching was considered in \cite{LongtermSSRA5}, where a deep reinforcement learning-based caching-aware handover strategy reduced failures and blocking under limited resources. 
In \cite{LongtermSSRA6}, a Lyapunov-based framework jointly optimized inter-satellite handover, beam hopping, and spectrum sharing to improve service satisfaction while mitigating interference. 

Different from the above studies, the proposed \methodname{} does not assume fixed constellation or static satellite amount.
\methodname{} is designed to handle dynamic visibility in \leoname{} systems by jointly optimizing satellite selection, UA, and RA under constellation size constraints, combining Markov approximation with matching game theory to tackle the challenges of time-varying candidate satellite sets and coupled decision variables.

\section{System Model and Problem Formulation}\label{System Model}
\subsection{Network Model}
Consider the downlink transmission of an \leoname{} satellite network where satellites distributed across $O$ orbital planes adaptively serve $J$ UEs within a designated area. 
Let $\mathcal{O}=\{1, \ldots, O\}$ and $\mathcal{U}=\{u_1, \ldots, u_J\}$ denote the sets of orbit indices and UEs, respectively.
Let $\mathcal{S}=\cup_{o=1}^O\mathcal{S}_o$ denote the set of satellites across $O$ orbital planes. $\mathcal{S}_o=\{ s_{o,1},\dots,s_{o,|\mathcal{S}_o|}\}$ represents the set of satellites on the $o\mbox{-th}$ orbital plane, and $|\mathcal{S}_o|$ denotes its cardinality.
The downlink spectrum is assumed to be partitioned into $K$ orthogonal and equal-bandwidth subcarriers that are represented by the set $\mathcal{B}=\{B_1, \ldots, B_K\}$. 

\subsection{Zenith Angle Rule}

The zenith angle of the satellites relative to UE is an important metric for satellite selection, as it directly influences the communication quality \cite{elevation_angle}.
To maintain line-of-sight (LoS) communication links, blockages caused by tall buildings should be avoided. 
A common method is to set a zenith-angle threshold when selecting satellites. 
Let $\varphi$ denote the maximum zenith angle at which a UE can observe satellites.
This threshold naturally defines a cone-shaped visibility region, and the satellites within this cone for $u_j$ constitute the candidate set $\mathcal{S}_{j}(t)$. 
The cone-shaped visibility region and the corresponding candidate set $\mathcal{S}_{j}(t)$ are illustrated in Fig. \ref{Cone_rule}. 
Based on the setup origin point, the coordinates of $u_j$, $s_{o,i}$, and the earth center are denoted as $\boldsymbol{a}_j(t)$, $\boldsymbol{a}_{o,i}(t)$, and $\boldsymbol{a}_{\text{e}}$, respectively. 
Accordingly, $\mathcal{S}_{j}(t)$ can be defined as
\begin{equation}
\mathcal{S}_{j}(t) = \{s_{o,i}\mid\langle\boldsymbol{a}_j(t)-\boldsymbol{a}_{\text{e}},\boldsymbol{a}_{o,i}(t)-\boldsymbol{a}_j(t) \rangle\leq \cos\varphi\},
\end{equation}
where $\langle \cdot,\cdot\rangle$ represents the operator of calculating the normalized inner product between two vectors.
\begin{figure}[tbp]
\centerline{\includegraphics[width=0.45\textwidth]{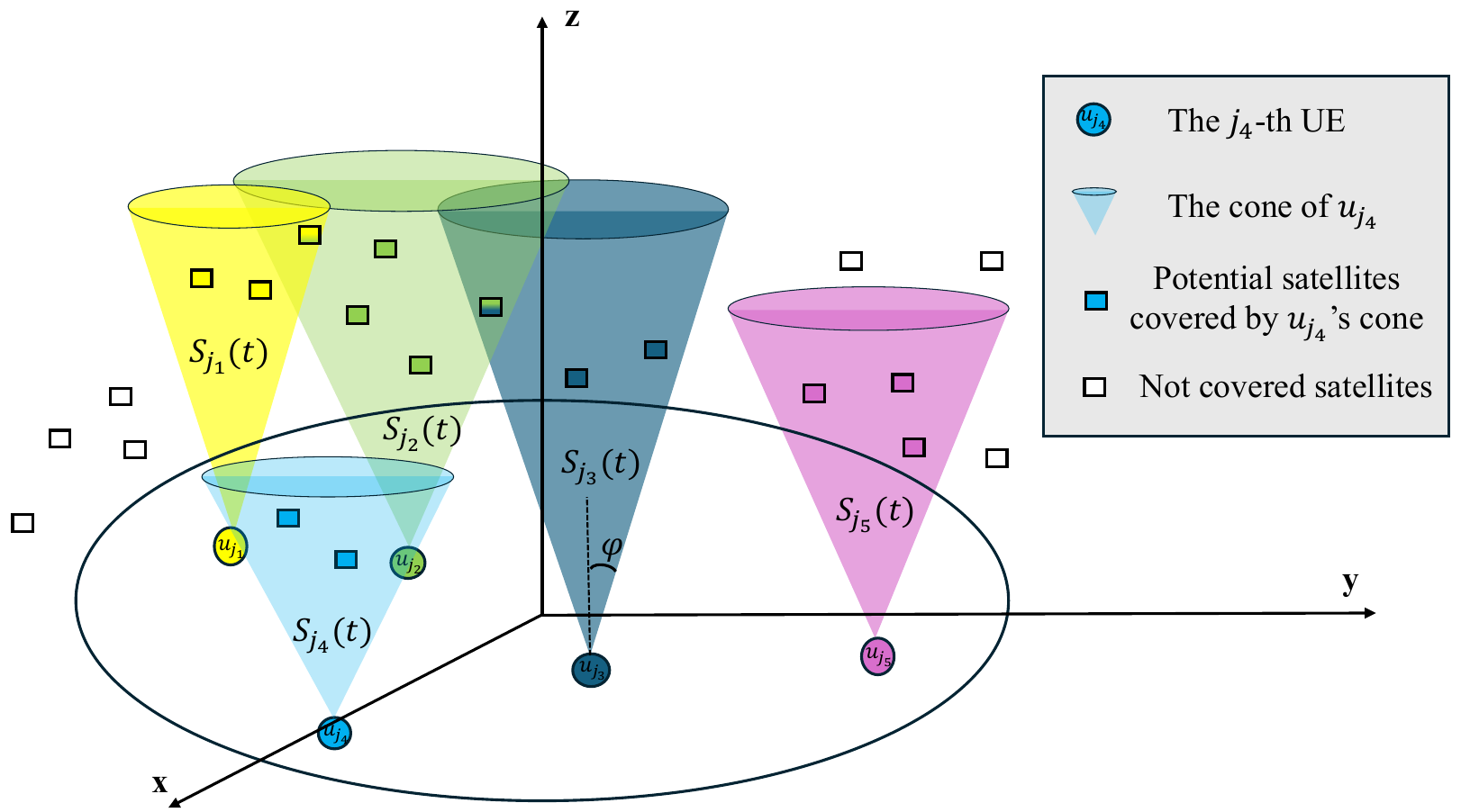}}
\caption{Illustration of the zenith angle rule.}
\label{Cone_rule}
\end{figure}
All satellites located within the individual cone regions collectively form the union set are denoted as $\tilde{\mathcal{S}}(t) = \cup_j \mathcal{S}_{j}(t)$. 
The serving satellites are then selected from $\tilde{\mathcal{S}}(t)$.
To indicate whether a specific satellite $s_{o,i}$ is actively providing service at time $t$, we introduce a binary decision variable $z_{o,i}(t) \in \{0,1\}$.  $z_{o,i}(t) = 1$ iff $s_{o,i}$ is selected; otherwise, $z_{o,i}(t) = 0$.
The final set of optimized satellites is denoted as $\tilde{\mathcal{S}}^*(t)\subseteq \tilde{\mathcal{S}}^(t)$.
Therefore, the set of satellites providing service at time slot $t$ is $\tilde{\mathcal{S}}^*(t)=\{s_{o,i}\mid z_{o,i}(t)=1\}.$

\subsection{Transmission Model}
It is important to note that only the satellites selected to provide service are eligible for BA. 
The set of subcarriers assigned to $s_{o,i}$ is denoted by $\mathcal{B}_{o,i}$. 
The sets fulfill the conditions $\mathcal{B} = \cup_{s_{o,i}\in \tilde{\mathcal{S}}}\mathcal{B}_{o,i}(t)$, and $|\mathcal{B}_{o,i}(t)|={K}/{|{\tilde{\mathcal{S}}^*(t)}|}$. 
We further assume that the frequency bands assigned to different service-providing satellites are mutually exclusive\cite{conference}. In other words, there is no overlap between $\mathcal{B}_{o,i}(t)$ and $\mathcal{B}_{o,i'}(t)$ for $i \neq i'$, i.e., $\mathcal{B}_{o,i}(t)\cap\mathcal{B}_{o',i'}(t) = \emptyset, \forall(o,i)\ne (o',i')|_{z_{o,i}=z_{o',i'}=1}.$

To characterize the connection relationships, we introduce two additional binary variables. 
The UA indicator is denoted as  $x_{o,i,j}(t) \in\{0,1\}$, where $x_{o,i,j}(t)=1$ if  $s_{o,i}$ serves $u_j$ at time slot $t$, and $x_{o,i,j}(t)=0$ otherwise. 
Likewise, the BA indicator is denoted as $y_{o,i,j}^k(t)\in\left\{0,1\right\}$ with $y_{o,i,j}^k(t)=1$ if frequency band $B_k$ is allocated to serve $u_j$, and $y_{o,i,j}^k(t)=0$ otherwise.
Hence, the set of UEs associated with $s_{o,i}$ is $\mathcal{U}_{o,i}(t)=\{u_j\mid x_{o,i,j}(t)=1\}$, while the set of UEs served by $B_k$ is denoted as $\mathcal{U}_{k}(t)=\{u_j\mid y_{o,i,j}^k(t)=1\}$.

The channel coefficient between $u_j$ and $s_{o,i}$ at time $t$ is 
\begin{equation}
H_{o,i,j}(t) = \sqrt{PL_{o,i,j}(t) \cdot SF(t)\cdot G_t}h_{o,i,j}(t),
\end{equation}
where $G_t$ denotes the transmitting antenna gain, $SF(t)$ and $h_{o,i,j}(t)$ represent the shadowing fading and the small scale fading channel coefficient between $s_{o,i}$ and $u_{j}$, respectively.
The path loss $PL_{o,i,j}(t)$ between $s_{o,i}$ and $u_j$ is given by
\begin{equation}
PL_{o,i,j}(t) = C_1\cdot \log_{10}d_{o,i,j}(t)+C_2+C_3\cdot \log_{10}f_c,
\end{equation}
where $d_{o,i,j}(t)$ represents the Euclidean distance between $s_{o,i}$ and $u_j$, and $f_c$ is the carrier frequency.  
The parameters $C_1$, $C_2$, and $C_3$ can be referred from tables provided in {\em QuaDRiGa} \cite{quadriga}.
The small scale fading is calculated as $h_{o,i,j}(t) = e^{j2\pi f_cd_{o,i,j}(t)}$.
The signal-to-interference-plus-noise ratio (SINR) experienced by $u_j$ on the $k\mbox{-th}$ subcarrier from its serving satellite $s_{o,i}$ at time $t$ is
\begin{equation}
\Gamma_{o,i,j}^{k}(t)=x_{o,i,j}(t)\cdot y_{o,i,j}^k(t)\frac{p_{j}(t)|H_{o,i,j}^{k}(t)|^2}{I_{o,i,j}^{k}(t)+\sigma^2},
\end{equation}
where $\sigma^2$ is the noise level,  $p_{j}(t)$ is $u_j$'s transmit power, and 
\begin{equation}\label{expr:Interfernce}
I_{o,i,j}^{k}(t) = \sum\nolimits_{j'\ne j}x_{o,i,j'}(t)y_{o,i,j'}^{k}(t)p_{j'}(t)|H_{o,i,j}^k(t)|^2,
\end{equation}
is the multiple access interference experienced by $u_j$ over $B_k$.
	
Therefore, the sum rate of the network can be written as 
\begin{equation}
R(t)=\sum\nolimits_{j=1}^{J} r_j(t),
\end{equation}
where $r_{j}(t)$ is the data rate of $u_j$ and is given by
\begin{equation}
r_{j}(t)=\sum\nolimits_{o,i}x_{o,i,j}(t)r_{o,i,j}(t),
\end{equation}
where $r_{o,i,j}(t)$ is the data rate of the communication link between $s_{o,i}$ and $u_j$ and is calculated as
\begin{equation}\label{expr:rj}
r_{o,i,j}(t)=\sum\nolimits_{k=1}^{K}y_{o,i,j}^{k}(t)B\cdot\log_2{(1+\Gamma_{o,i,j}^{k}(t))}.
\end{equation}
	
\subsection{Problem Formulation}
As mentioned in Section III.B, our setting explicitly accounts for dynamic visibility induced by each UE's cone-shaped field of view. 
As satellites on different orbital planes and relative phases enter/exit these cones, each per-UE candidate set ${\mathcal S}_j(t), \forall u_j\in{\mathcal U}$, along with their union $\tilde{\mathcal S}(t)$, evolves across time slots. 
Hence, service satellites must be adapted online under constellation size limits and per-satellite power budgets. 
This time-varying feasible region tightly couples satellite selection, UA and PA, and substantially enlarges the combinatorial search space. 
Consequently, our problem includes both dynamic visibility and constellation size constraints, leading to the formulation of an optimization problem that simultaneously addresses satellite selection, UA, and RA, as below

    \begin{align}
\textbf{OP-0}: &\max_{\boldsymbol{x},\boldsymbol{y},\boldsymbol{z},\boldsymbol{p}}  \quad \sum\nolimits_{t\in [0,T]}R(t)\label{obj}\\
\quad \text { s.t. } & \sum\nolimits_{o,i}x_{o,i,j}(t)\cdot z_{o,i}(t) \leq 1, \quad \forall j, \label{constraint_association}\\
& \sum\nolimits_{j=1}^{J}x_{o,i,j}(t)\cdot z_{o,i}(t)\cdot p_j(t) \leq P_{\text{max}}, \quad \forall o,i, \label{constraint_pmax}\\
& \sum\nolimits_{o,i} z_{o,i}(t) \leq Z_{th} \label{constraint_SatNum},\\
& r_{j}(t) \geq r_{\text{min},j}(t), \quad \forall j,\label{constraint_rmin}\\
& p_j(t)>0, \quad \forall j,\label{constraint_pj}\\
& s_{o,i}(t)\in \tilde{\mathcal{S}}(t),\quad \forall (o,i)|_{z_{o,i}=1}, \label{constraint_cone}\\
& x_{o,i,j}(t), y_{o,i,j}^{k}(t), z_{o,i}(t)\in \left\{0,1\right\}, \quad \forall o,i,j,k,\label{constraint_xyzbin}
\end{align}    
where \eqref{constraint_association} enforces that each UE can be associated with at most one satellite.
\eqref{constraint_pmax} ensures that the transmit power of each satellite does not exceed its maximum limit. 
\eqref{constraint_SatNum} represents the constellation size constraint. Therefore the number of the chosen satellites should not be larger than $Z_{th}$. 
\eqref{constraint_rmin} requires that the minimum rate demand of every UE is satisfied. 
\eqref{constraint_pj} specifies that the allocated power to each UE must remain positive.
\eqref{constraint_cone} guarantees that the service satellites are selected from the union set defined by the cone-based visibility rule.

Problem \textbf{OP-0} has integer constraints \eqref{constraint_association}-\eqref{constraint_SatNum} and \eqref{constraint_cone}-\eqref{constraint_xyzbin}, as well as continuous constraints \eqref{constraint_pmax} and \eqref{constraint_rmin}-\eqref{constraint_pj}. The problem  is non-convex and mixed integer programming with respect to $\{\boldsymbol{x},\boldsymbol{y},\boldsymbol{z},\boldsymbol{p}\}$, where ${\boldsymbol x}=[x_{o,i,j}]_{\forall o, i,j}$, ${\boldsymbol y}=[y^k_{o,i,j}]_{\forall o, i,j, k}$ and ${\boldsymbol z}=[z_{o,i}]_{\forall o, i}$, ${\boldsymbol p}=[p_{j}]_{\forall j}$.
In the absence of efficient solvers for the problem, we adopt a tractable alternative. 
Specifically, we will use a problem-specific Markov approximation that exploits the combinatorial structure for sub-optimal solutions with manageable complexity as follows.

\section{The Proposed Solution Framework}\label{Algorithm Design}
To effectively solve problem \textbf{OP-0}, we propose \methodname{} framework which integrates four complementary components into a unified solution pipeline. 
First, Markov approximation is applied to address the combinatorial nature of satellite selection under constellation size constraints and dynamic visibility variations. Then we reformulate the problem into a tractable form with guaranteed convergence to a stationary distribution. 
Next, based on the selected satellite set, a BCD structure is introduced to solve UA, BA, and PA alternatively. Specifically, two matching games are utilized to establish stable UA with BA. Conditioned on the outcome of UA and BA, PA is optimized using KKT conditions.
By combining these components in a sequential yet coherent manner, \methodname{} becomes a scalable and suboptimal framework for the joint satellite selection and joint UA, BA, and PA.
 The detail of each component will be elaborated in the following.

\subsection{Markov Approximation}
Let $\boldsymbol{f}=\{\boldsymbol{x},\boldsymbol{y},\boldsymbol{z}, \boldsymbol{p}\}\in \mathcal{F}$ denote a feasible \leoname{} satellites network configuration with $\mathcal{F}$ as the feasible set generated by the constraints (\ref{constraint_association})-(\ref{constraint_xyzbin}).
The objective function in \eqref{obj} can be reevaluated as $\theta(\boldsymbol{f}) = \sum_{t\in [0,T]}R(t)$ with $\boldsymbol{f}\in\mathcal{F}$. 
Accordingly, the original objective function can be equivalently reformulated as
\begin{equation}\label{eq:discreteproblem}
\textbf{OP-1}: \max\nolimits_{\boldsymbol{f}\in \mathcal{F}}\quad \theta(\boldsymbol{f}).
\end{equation}
$\textbf{OP-1}$ is inherently discrete and combinatorial. To make it tractable, we approximate the problem by reformulating it as a continuous optimization problem through log-sum-exp relaxation and the use of a conjugate function \cite{markov} as

\begin{align}
\textbf{OP-2}: \quad \max_{\boldsymbol{w}}\quad &\sum_{\boldsymbol{f}_i\in \mathcal{F}}w_{\boldsymbol{f}_i}\theta(\boldsymbol{f}_i)-\frac{1}{\beta}\sum_{\boldsymbol{f}_i\in \mathcal{F}}w_{\boldsymbol{f}_i}\log(w_{\boldsymbol{f}_i})\label{P2}\\
\text{s.t.} \quad & \sum\nolimits_{\boldsymbol{f}_i\in \mathcal{F}}w_{\boldsymbol{f}_i} = 1,
\end{align}
where $w_{\boldsymbol{f}_i}$ denotes the weight of a feasible network configuration $\boldsymbol{f}_i$, and $\beta$ is a temperature parameter controlling the difference between the objective function value of $\textbf{OP-1}$ and that of $\textbf{OP-2}$. 
KKT condition for $\textbf{OP-2}$ leads to the optimal weights as
\begin{equation}\label{KKT}
w_{\boldsymbol{f}_i}^* = \frac{e^{\beta\theta(\boldsymbol{f}_i)}}{\sum_{m=1}^{|{\mathcal{F}}|}e^{\beta\theta(\boldsymbol{f}_m))}}.
\end{equation}

However, it is not practical to compute the optimal weights for all feasible network configurations. 
To address this challenge, we regard the feasible configurations as states and the corresponding weight as the stationary probability of a Markov chain. 
The problem then becomes the design of a Markov chain with $\mathcal{F}$ as the state space and the stationary distribution as (\ref{KKT}). We note that there is at least one Markov chain fulfilling our requirements\cite{markov}.
Let  $q(\boldsymbol{f}\rightarrow \boldsymbol{f}')$ and $q(\boldsymbol{f}'\rightarrow \boldsymbol{f})$ denote the transition probabilities. 
The Markov chain should satisfy the balance condition and is expressed as
\begin{equation}
w_{\boldsymbol{f}}q(\boldsymbol{f}\rightarrow \boldsymbol{f}')=w_{\boldsymbol{f}'}q(\boldsymbol{f}'\rightarrow\boldsymbol{f}),\label{eq:equality_of_joint}
\end{equation}
which is derived according to the Bayes rule on a joint probability of $\boldsymbol{f}$ and $\boldsymbol{f}'$. After substituting \eqref{KKT} into \eqref{eq:equality_of_joint} and canceling out the denominators on both sides, we can obtain
\begin{equation}
e^{\beta\theta(\boldsymbol{f})}q(\boldsymbol{f}\rightarrow \boldsymbol{f}') = 	e^{\beta\theta(\boldsymbol{f}')}q(\boldsymbol{f}'\rightarrow \boldsymbol{f}).\label{eq:equality_of_joint2}
\end{equation}
We assume $q(\boldsymbol{f}\rightarrow \boldsymbol{f}')+q(\boldsymbol{f}'\rightarrow \boldsymbol{f})=1$ in Markov approximation\cite{markov}. The transition probabilities can be evaluated as
\begin{align}
q(\boldsymbol{f}\rightarrow \boldsymbol{f}')&=\frac{1}{1+e^{\beta(\theta(\boldsymbol{f})-\theta(\boldsymbol{f}'))}}, \label{eq:prob_trans1}\\
q(\boldsymbol{f}'\rightarrow \boldsymbol{f})&=\frac{1}{1+e^{\beta(\theta(\boldsymbol{f}')-\theta(\boldsymbol{f}))}}.\label{eq:prob_trans2}
\end{align}
\textcolor{black}{With the state space specified and the transition probabilities given in \eqref{eq:prob_trans1}--\eqref{eq:prob_trans2}, we summarize the Markov-approximation procedure in Algorithm \ref{alg:markovapproximation}. The algorithm iteratively alternates between an exploration stage and a consolidation stage. In each iteration, Algorithm~\ref{alg:JUBPA} is invoked to determine UA, BA, and PA based on the currently selected satellites. If the state stabilizes (i.e., remains unchanged across iterations), the exploration probability $\nu^{[c_{\text{o}}]}$ will be gradually reduced. Algorithm~\ref{alg:markovapproximation} terminates once the exploration probability decreases to zero, at which point a stable solution is obtained.}

\begin{algorithm}[t]
		\caption{The \methodname{} Framework}\label{alg:markovapproximation}
		\begin{algorithmic}[1]
			\STATE \textbf{Initialization:} $c_{\text{o}}=0$, $\nu^{[0]}=1$, $\beta_{\text{step}}>0$, $\nu_\text{step}>0$, $\boldsymbol{z}^{[0]}$ drawn at random.
			\WHILE {$\nu^{[c_{\text{o}}]}\ne0$}
			\IF {$\mod(c_{\text{o}},2)=0$}
			\STATE \% Exploration stage
			\STATE $\boldsymbol{z}^{[c_{\text{o}}+1]}=
			\begin{cases}
				\boldsymbol{z}\quad \text{drawn randomly}, &\text{with prob. }\nu^{[c_{\text{o}}]}\\
				\boldsymbol{z}^{[c_{\text{o}}]}, &\text{with prob. }1-\nu^{[c_{\text{o}}]}
			\end{cases}$
			\ELSE
			\STATE \% Consolidation stage
			\STATE Calculate $\alpha=q(\boldsymbol{f}^{[c_{\text{o}}]}\rightarrow \boldsymbol{f}^{[c_{\text{o}}-1]})$ based on \eqref{eq:prob_trans1}.
			\STATE $\boldsymbol{f}^{[c_{\text{o}}+1]}=
			\begin{cases}
				\boldsymbol{f}^{[c_{\text{o}}-1]}, &\text{with prob. }\alpha\\
				\boldsymbol{f}^{[c_{\text{o}}]}, &\text{with prob. }1-\alpha
			\end{cases}$
			\STATE $\beta=\beta+\beta_{\text{step}}$
			\IF {$\boldsymbol{z}^{[c_{\text{o}}+1]} = \boldsymbol{z}^{[c_{\text{o}}]} $}
			\STATE $\nu^{[c_{\text{o}}+1]}=\max\left\{0,\nu^{[c_{\text{o}}]}-\nu_{\text{step}}\right\}$
			\ENDIF 
			\ENDIF
			\STATE $\boldsymbol{f}^{[c_{\text{o}}+1]*}=\left\{\boldsymbol{x}^{[c_{\text{o}}+1]*},\boldsymbol{y}^{[c_{\text{o}}+1]*},\boldsymbol{p}^{[c_{\text{o}}+1]*},\boldsymbol{z}^{[c_{\text{o}}+1]}\right\}$ by running \textbf{\stagename{}} (Algorithm \ref{alg:JUBPA})
			\STATE $c_{\text{o}}=c_{\text{o}}+1$
			\ENDWHILE
			\STATE \textbf{Return} $\boldsymbol{z}^*$,$\boldsymbol{x}^*$,$\boldsymbol{y}^*$,$\boldsymbol{p}^*$
		\end{algorithmic}
	\end{algorithm}
	
	\begin{algorithm}[tb]
		\caption{Algorithm of JUBPA}\label{alg:JUBPA}
		\begin{algorithmic}[1]
			\ENSURE{$\boldsymbol{x}^{[c_{\text{o}}+1]*},\boldsymbol{y}^{[c_{\text{o}}+1]*},\boldsymbol{p}^{[c_{\text{o}}+1]*}$}
			\REQUIRE{$\boldsymbol{z}^{[c_{\text{o}}+1]}$} 

			\FOR{$u_j$ in $\boldsymbol{L}_{\text{UE}}^{\text{init}}$}
			\FOR{$B_k$ in $\boldsymbol{L}_j^{\text{init}}$} 
			\IF{$G^{+}(G_k,u_j)$ is solvable}
			\STATE $\mathcal{U}_k\leftarrow \mathcal{U}_k\cup \left\{u_j\right\}$, Break
			\ELSE
			\STATE Continue
			\ENDIF
			\ENDFOR
			\ENDFOR
			\WHILE {$(\boldsymbol{x},\boldsymbol{y})$ is not stable}
			\STATE \textbf{\textit{UA+BA Phase:}}
			Run \textbf{Algorithm \ref{alg:UABA}} to obtain $\boldsymbol{x}^{[c_{\text{o}}+1,c_{\text{i}}]}$ and $\boldsymbol{y}^{[c_{\text{o}}+1,c_{\text{i}}]}$

			\STATE \textbf{\textit{PA Phase:}}
			Run \textbf{Algorithm \ref{alg:PA}} to obtain $\boldsymbol{p}^{[c_{\text{o}}+1,c_{\text{i}}]}$
			\STATE $c_{\text{i}}=c_{\text{i}}+1$
			\ENDWHILE
			\STATE \textbf{Return} $\boldsymbol{x}^{[c_{\text{o}}+1]*},\boldsymbol{y}^{[c_{\text{o}}+1]*},\boldsymbol{p}^{[c_{\text{o}}+1]*}$
		\end{algorithmic}
	\end{algorithm}
	
Before executing Algorithm \ref{alg:JUBPA}, it is necessary to find a feasibly initialized solution point. For the purpose, we will find a solution point at the boundary line generated by the minimum rate constraint (\ref{constraint_rmin}). Thus, we let the communication rate of each $u_j$ equal to $r_{\text{min},j}$.

\textcolor{black}{After substituting (\refeq{expr:rj}) into the equation $r_{j}(t) = r_{\text{min},j}(t)$ }and omitting the indexes $o,i,t$ for simplicity, $r_{\text{min},j}$ can be evaluated as 
\begin{equation}\label{eq:1}
r_{\text{min},j} = B\cdot \log_2\left[1+\frac{p_j|H_{j}^k|^2}{\sum_{j',y_{j'}^k=1}p_{j'}\cdot |H_{j}^k|^2+\sigma^2}\right], \forall j,
\end{equation}
which is equivalent to 
\begin{equation}\label{expr: rj2}
p_j\cdot \frac{1}{\delta_j}-\sum\nolimits_{j',y_{j'}^k=1}p_{j'}=\epsilon_{j,k}, \quad \forall j,
\end{equation}
where  $\delta_j = 2^{\frac{r_{j,\text{min}}}{B}}-1$, $\epsilon_{j,k} = \frac{\sigma^2}{|H_{j}^k|^2}$. From (\ref{expr: rj2}), for each subcarrier $B_k$ we obtain a system of the first-degree linear equations as
\begin{equation}
G_k=G(\mathcal{U}_k):\mathbf{A}_k\boldsymbol{p}_k=\boldsymbol{\epsilon}_k,\forall k,
\end{equation}
where $\mathcal{U}_k$ denotes the UE set served by $B_k$, $G(\mathcal{U}_k)$ denotes the linear equations system induced by the UE set $\mathcal{U}_k$. The UE power vector $\boldsymbol{p}_k\in \mathbb{R}^{J_k\times 1}$ is defined as $\left[\boldsymbol{p}_k\right]_{j_k} = p_{j_k}$, and $\boldsymbol{\epsilon}_k\in \mathbb{R}^{J_k\times 1}$ is defined as $\left[\boldsymbol{\epsilon}_k\right]_{j_k} = \epsilon_{j_k}^k$. $J_k=\sum_{j_k}y_{j_k}^k$ is the number of UEs served by $B_k$. $\mathbf{A}_k\in \mathbb{R}^{J_k\times J_k}$ is defined as 
	\begin{equation}
		\left[\mathbf{A}_k\right]_{ij}=
		\begin{cases}
			\frac{1}{\delta_j}, & i=j,\\
			-1,& i\ne j.
		\end{cases}
	\end{equation}
    
To model the incremental augmentation of $\mathcal{U}_k$ when $u_j$ is assigned to $B_k$, we define the augmentation operator as
\begin{equation}
	G^{+}(G_k,u_j): \mathbf{A}_{j\rightarrow k}\boldsymbol{p}_{j\rightarrow k}=\boldsymbol{\epsilon}_{j\rightarrow k},
\end{equation}
which maps the previous linear system of $B_k$ and the candidate UE $u_j$ to the updated linear system after augmentation. $\boldsymbol{p}_{j\rightarrow k}\in \mathbb{R}^{(J_k+1)\times 1}$ is defined as $\boldsymbol{p}_{j\rightarrow k} = [\boldsymbol{p}_{k}^{T},p_j]^T$, and $\boldsymbol{\epsilon}_{j\rightarrow k}\in \mathbb{R}^{(J_k+1)\times 1}$ is defined as $\boldsymbol{\epsilon}_{j\rightarrow k} = [\boldsymbol{\epsilon}_k^{T},\epsilon_{j,k}]^T$. The updated matrix $\mathbf{A}_{j\rightarrow k}\in \mathbb{R}^{(J_k+1)\times(J_k+1)}$ is defined as 
	\begin{equation}
		\left[\mathbf{A}_{j\rightarrow k}\right]_{i'j'}=
		\begin{cases}
			\left[\mathbf{A}_k\right]_{i'j'}, & i',j'\leq J_k,\\
			{1}/{\delta_j}, & i'=j'=J_k+1,\\
			-1.& \text{others}
		\end{cases}
	\end{equation}

	The definition of solvable $G^{+}(G_k,u_j)$ is 
	\begin{align}
		G^{+}(G_k,u_j)\text{ is solvable}\Longleftrightarrow \det\left(\mathbf{A}_{j\rightarrow k}\right)&\ne0,\nonumber \\
		\left(\mathbf{A}_{j\rightarrow k}\right)^{-1}\boldsymbol{\epsilon}_{j\rightarrow k} &>\boldsymbol{0}.
	\end{align}
	In the initialization of \stagename{},  the UE with higher $r_{\text{min},j}/\bar{g}_j(t)$ is prioritized to be allocated to the subcarrier with the highest channel gain. \textcolor{black}{This is because the UE with higher $r_{\text{min},j}/\bar{g}_j(t)$ has more stringent minimum rate requirements but experiences low channel gains. As a result, satisfying the constraint on the minimum rate for such UEs is challenging.}

\subsection{Matching Game Theory}
\textcolor{black}{
With a feasible solution from the initialization process, we will optimize  UA and BA with the matching game method,
\begin{align}
    \textbf{OP-UABA}: &\max_{\boldsymbol{x},\boldsymbol{y}}  \quad R(t)\\
    \quad \text { s.t. }& \text{(\ref{constraint_association}),~(\ref{constraint_pmax}),~(\ref{constraint_rmin}),~(\ref{constraint_xyzbin})}.
\end{align}
In the following, we design two matching games corresponding to UA and BA, respectively. Both matching games are designed based on the explored $\boldsymbol{z}^{[c_{\text{o}}+1]}$ and the corresponding satellites set is defined as $\mathcal{S}_{\boldsymbol{z}^{[c_{\text{o}}+1]}}(t)$ which will be denoted as $\mathcal{S}_{\boldsymbol{z}}(t)$ for simplicity in the following. There are two sides of players in each matching game. In the UA matching game, two sides are UEs and satellites, and the corresponding sets are $\mathcal{U}$ and $\mathcal{S}_{\boldsymbol{z}}(t)$. In the BA matching game, two sides are UEs associated with $s_{o,i}$ and subcarriers belonging to $s_{o,i}$, and the corresponding sets are $\mathcal{U}_{o,i}$ and $\mathcal{B}_{o,i}$. The matching function is used to describe the mapping relationship between two sides of players in each game.} The matching function $\mu_{\text{UA}}(\cdot)$ for UA is defined as a function mapping $\mathcal{U}$ to $\mathcal{S}_{\boldsymbol{z}}(t)$ such that $s_{o,i} = \mu_{\text{UA}}(u_j)$ iff $u_j \in \mu^{-1}_{\text{UA}}(s_{o,i})$, in which $\mu^{-1}_{\text{UA}}(s_{o,i})$ is defined as 
\begin{equation}
	\mu^{-1}_{\text{UA}}(s_{o,i};t)=\mathcal{U}_{o,i}=\left\{u_j\mid x_{o,i,j}(t)=1\right\}.
\end{equation}
Likewise, the matching $\mu_{\text{BA}}(\cdot)$ for BA is defined as a function mapping $\mathcal{U}_{o,i}$ to $\mathcal{B}_{o,i}$ such that $B_{k} = \mu_{\text{BA}}(u_j)$ iff $u_j \in \mu^{-1}_{\text{BA}}(B_{k})$, in which $\mu^{-1}_{\text{BA}}(B_k)$ is defined as
\begin{equation}
	\mu^{-1}_{\text{BA}}(B_k;t)=\mathcal{U}_{k}=\{u_j\mid y^k_{o,i,j}(t)=1\}.
\end{equation}

The final matching result of UA and BA is denoted as $\boldsymbol{\mu}=\{{\mu_{\text{UA}},\mu_{\text{BA}}}\}$. After defining basic elements in a matching game, how to design a matching game to obtain a stable matching result is introduced as the following.

The general procedure of a matching game can be summarized as three steps:
\begin{enumerate}
    \item Both sides of players construct the profile lists based on the preference function;
    \item  Each player on the first side proposes the requirement of being associated with the player corresponding to the highest preference function on the second side;
    \item Each player on the second side decides whether to be associated with the first-side player that proposed the association requirement based on the profile list.
\end{enumerate}
In order to apply matching game to solve the problem \textbf{OP-UABA}, we need to design the preference functions that are related to the objective function and constraints in \textbf{OP-UABA}. In the UA matching game, the preference function on the UE side is designed based on the objective function, and the preference function on satellites side is designed based on the constraints. In the BA matching game, the preference function is designed based on objective function. The constraints are handled in a single step where we check whether they are satisfied. The pseudo code is shown in Algorithm \ref{alg:UABA}.

\subsubsection{UA} 
    \begin{algorithm}[tb]
		\caption{Algorithm of UA and BA}\label{alg:UABA}
		\begin{algorithmic}[1]
			\ENSURE{$\boldsymbol{x}^{[c_{\text{o}}+1, c_{\text{i}}]}$,$\boldsymbol{y}^{[c_{\text{o}}+1, c_{\text{i}}]}$}
			\REQUIRE{$\boldsymbol{x}^{[c_{\text{o}}+1, c_{\text{i}}-1]},\boldsymbol{y}^{[c_{\text{o}}+1, c_{\text{i}}-1]},\boldsymbol{p}^{[c_{\text{o}}+1, c_{\text{i}}-1]}$} 
			\STATE \% UA part
            \FOR {each $u_j$}
			\STATE Update $\boldsymbol{L}_{j}^{\text{UA},[c_{\text{in}}]}$ based on  $PF_{j}^{\text{UA}}(o,i;\boldsymbol{\mu}^{[c_{\text{in}}-1]})$
			
			\ENDFOR
			\STATE Generate $\boldsymbol{L}^{\gamma}(t)$ and $\mathcal{U}_{\text{req}}^{\text{UA}}$
			\WHILE {$\mathcal{U}_\text{req}^{\text{UA}} \ne \emptyset$}
			\FOR {$u_j\in \mathcal{U}_\text{req}^{\text{UA}}$}
			\IF{$U_j\leq U_{\text{max}}$ and $\mu_{\text{UA}}(u_j)\ne s_{o^{*},i^{*}}$}
			\STATE $o^{*},i^{*} = \mathop{\arg\max}_{s_{o,i}\in \boldsymbol{L}_{j}^{\text{UA}}}PF^{\text{UA}}_{j}(o,i)$,$b_{j\rightarrow o^{*},i^{*}}= 1$
			\ENDIF
			\ENDFOR
			\FOR {each $s_{o,i}\in \mathcal{S}_{\boldsymbol{z}}$}
			\STATE $U_{o,i}^{\text{req}} = \{u_j\mid b_{j\rightarrow (o,i)}= 1\}$,$p_{o,i}^{\text{req}}=\sum_{u_j\in \mathcal{U}_{o,i}^{\text{req}}}p_j$
		
			\IF {$p_{o,i}+p_{o,i}^{\text{req}}\leq p_{\text{max}}$}
			\STATE $\mathcal{U}_{o,i}\leftarrow \mathcal{U}_{o,i}^{\text{req}}\cup \mathcal{U}_{o,i}$, $p_{o,i}+=p_{o,i}^{\text{req}}$
			\STATE $\mathcal{U}_{o,i}^{\text{req}}\leftarrow \emptyset$, $\mathcal{U}_\text{req}^{\text{UA}}\leftarrow \mathcal{U}_\text{req}^{\text{UA}}\setminus \{u_{j}\mid u_{j}\in \mathcal{U}_{\text{req}}^{\text{UA}}\}$
			
			\ELSE

			\WHILE{$p_{o,i}< p_{\text{max}}$}
			\STATE $j' = \mathop{\arg\max}_{u_j\in \boldsymbol{L}_{o,i}^{\text{UA}}}\sum_{j\in \mathcal{U}_{o,i}}PF^{\text{UA}}_{o,i}(j)$
			\STATE$\mathcal{U}_{o,i}\leftarrow \mathcal{U}_{o,i}\cup \{u_{j'}\}$, $p_{o,i}\leftarrow p_{o,i}+p_{j'}$
			\STATE $\mathcal{U}_\text{req}^{\text{UA}}\setminus \{u_{j'}\}$,$\mathcal{U}_{o,i}^{\text{req}}\setminus \{u_{j'}\}$,$\boldsymbol{L}_{o,i}^{\text{UA}}\setminus \{u_{j'}\}$

			\ENDWHILE
			\ENDIF
			
			\FOR {$u_j \in \mathcal{U}_{o,i}^{\text{rej}}$}
			\STATE $\boldsymbol{L}_{j}^{\text{UA},[c_{\text{in}}]}\leftarrow \boldsymbol{L}_{j}^{\text{UA},[c_{\text{in}}]}\setminus \{s_{o,i}\}$
			\ENDFOR
			\ENDFOR
			
			\ENDWHILE
        \STATE \% BA part
        \FOR {$s_{o,i}\in \mathcal{S}^{\text{UA},[c_{i}]}$}
		
		\STATE The list $\boldsymbol{L}^{\text{UA},[c_i]}_{o,i}=\operatorname{argsort}_{u_j \in \mathcal{U}^{\text{UA},[c_{i}]}_{o,i}}p_j(t)$
		\FOR {$u_j\in \boldsymbol{L}^{\text{UA},[c_i]}_{o,i}$}
		\STATE The sorted list $\boldsymbol{L}^{o,i}_{j}=\operatorname{argsort}_{B_k \in \mathcal{B}_{o,i}} g_j^k(t)$
		\FOR{$B_k \in \boldsymbol{L}^{o,i}_{j}$}
		\IF {constraint (\ref{constraint_rmin}) in $B_{k^{*}}$ is fulfilled}
		\STATE $y_{o,i,j}^k(t)=1$
		\ELSE 
		\STATE $x_{o,i,j}(t)=0, x_{o,i,j}^{[c_i-1]}=1, U_j-=1$
		\ENDIF
		\ENDFOR
		\ENDFOR 
		\ENDFOR
            \FOR {each $s_{o,i}\in \mathcal{S}_{\boldsymbol{z}}$}
			\WHILE {$\max PF^{\text{BA}}(j,k)> 0$}
			\STATE $\mathcal{U}_{\text{req}}^{\text{BA}}\leftarrow \emptyset$,$\mathcal{B}_{\text{block}} \leftarrow \emptyset$
			\FOR {each $B_k\in \mathcal{B}_{o,i}$}
			
			\STATE $\mathcal{U}_{\text{req}}^{\text{BA}}\leftarrow \mathcal{U}_{\text{req}}^{\text{BA}}\cup \left\{\arg\min_{u_j\in \mathcal{U}_k}r_j(t)\right\}$
			\ENDFOR
			
			\WHILE {$\max \boldsymbol{L}^{\text{BA}}\ge 0$}
			\STATE $(j^{*},k^{*}) = \mathop{\arg\max}_{j,k\in \boldsymbol{L}^{\text{BA}}}PF^{\text{BA}}(j,k)$ 
			\IF {$\left\{\mu^{\text{BA}}(j^{*}), B_{k^{*}}\right\}\notin \mathcal{B}_{\text{block}}$}
			\IF {constraint (\ref{constraint_rmin}) in $B_{k^{*}}$ is fulfilled}
			\STATE $\mathcal{U}_{k^{*}}\leftarrow \mathcal{U}_{k^{*}}\cup u_{j^{*}}$,$\mathcal{U}_{\mu^{\text{BA}}(j^{*})}\leftarrow \mathcal{U}_{\mu^{\text{BA}}(j^{*})}\setminus j^{*}$
			\STATE $\mathcal{B}_{\text{block}}\leftarrow \mathcal{B}_{\text{block}}\cup \left\{\mu^{\text{BA}}(j^{*}), B_{k^{*}}\right\}$
			\ENDIF
			\STATE $\boldsymbol{L}^{\text{BA}}\leftarrow \boldsymbol{L}^{\text{BA}}\setminus \left\{(j^{*},k^{*})\right\}$
			\ENDIF
			\ENDWHILE
			\ENDWHILE
			\ENDFOR
		\end{algorithmic}
	\end{algorithm}
    \textcolor{black}{As mentioned above, the preference function on the UE side is designed based on the objective function value. Thus we define the preference function as the potential average increment of each UE's communication rate after changing the associated satellite.} 
    
    The rate of $u_j$ when it is served by $B_k$ is defined as
\begin{align}
&r_{j\to k}(t)= y_j^{k}(t)r_j(t)+\nonumber\\
&(1 - y_j^{k}(t))B\log_2\!\left(1+\frac{p_j(t)\,|H_{j}^{k}(t)|^{2}}{\displaystyle \sum\limits_{j'\in \mathcal{U}_k} p_{j'}(t)\,|H_{j}^{k}(t)|^{2} + \sigma^{2}}\right).
\end{align}
The set of subcarriers belonging to $s_{o,i}$ which can make $r_{j\to k}(t)$ higher than $r_j(t)$ is denoted as
	\begin{equation}
		\mathcal{B}^{+}_{j\rightarrow o,i}=\left\{B_k\mid r_{j\rightarrow k}(t)\geq r_j(t),B_k\in \mathcal{B}_{o,i}\right\}.
	\end{equation}
	The profile function of $u_j$ is defined as
	\begin{equation}
		PF_{j}^{\text{UA}}(o,i;\boldsymbol{\mu},t) = \frac{1}{|\mathcal{B}^{+}_{j\rightarrow o,i}|}\sum_{k\in \mathcal{B}^{+}_{j\rightarrow o,i}}\left(r_{j\rightarrow k}(t)-r_j(t)\right).
	\end{equation}
	Therefore, if the subcarriers of $s_{o,i}$ can make the communication rate of $u_j$ higher after serving $u_j$, $s_{o,i}$ will correspond to a higher preference function value. The comparison operator for $u_j$ at UA stage is defined as
	\begin{equation}\label{eq:UA_comparison}
		\begin{split}
			&(s_{o,i},\mu_{\text{UA}})\succ_{j;t}^{\text{UA}}(s_{o',i'},\mu'_{\text{UA}})\Leftrightarrow \\
			&PF_{j}^{\text{UA}}(o,i;\boldsymbol{\mu},t)\ge PF_{j}^{\text{UA}}(o',i';\boldsymbol{\mu}',t).
		\end{split}
	\end{equation}
	
The profile list of $u_j$ is denoted as $\boldsymbol{L}^{\text{UA}}_{j}(t)$, representing a list of satellites. $\boldsymbol{L}^{\text{UA}}_{j}(t)$ is sorted in descending order according to \eqref{eq:UA_comparison}, i.e.,
   
\begin{equation}
\boldsymbol{L}^{\text{UA}}_{j}(t)=\underset{{s_{o,i}\in \mathcal{S}_{\boldsymbol{z}}}}{\operatorname{argsort}}\quad PF_{j}^{\text{UA}}(o,i),
\end{equation}
The maximum average achievable rate of $u_j$ can be derived as $\gamma_j(t)=\max\nolimits_{s_{o,i}\in \boldsymbol{L}^{\text{UA}}_{j}(t)}PF_{j}^{\text{UA}}(o,i).$

To avoid oscillatory UA updates, we introduce two stabilization mechanisms. The first one is per-iteration quota on UA changes. Because one UE’s association change perturbs the interference of others, allowing UEs to change associations concurrently can induce large network-wide fluctuations and hinder convergence. Hence, in each iteration only a quota-limited subset of UEs is permitted to propose a change. Define this subset of UEs as
\begin{equation}
\boldsymbol{L}^{\gamma}(t)=\underset{{u_j\in \mathcal{U}}}{\operatorname{argsort}}\quad \gamma_j(t).
\end{equation}
where $\gamma_j(t)=\max\nolimits_{s_{o,i}\in \boldsymbol{L}^{\text{UA}}_{j}(t)}PF_{j}^{\text{UA}}(o,i)$. Given a quota $U_q$, only the first $U_q$ UEs in $\boldsymbol{L}^{\gamma}(t)$ are allowed to request a UA update. This set is denoted as $\mathcal{U}_{\text{req}}^{\text{UA}}$ and is defined as
\begin{equation}
\mathcal{U}_{\text{req}}^{\text{UA}}=\{u_{j}\in \boldsymbol{L}^{\gamma}(t)\mid 1\leq j \leq U_q\}.
\end{equation}

The second mechanism is per-UE limit on the number of UA changes. To prevent back-and-forth switching by any single UE, we impose a hard limit on its total number of association updates. Let $U_j$ be the cumulative count of $u_j$ changing UA. A proposal from $u_j$ is accepted only if $U_j < U_{\text{max}}$. $U_{\text{max}}$ is an integer limit.

After UE proposing the request of changing association, satellites will decide whether to accept the requests according to their preference function. To define the preference function on satellites' side, from (\ref{constraint_rmin}), the minimum-rate constraint of $u_j$ on subcarrier $B_k$ can be equivalently written as

\begin{align}
{p_{j}(t)}/{\delta_j}-\epsilon_{j,k}\ge \sum\nolimits_{\substack{j'\ne j,\mu(u_{j'};t)=B_{k}}} p_{j'}(t).\label{UA_cons}
\end{align}
We denote the left-hand term as
\begin{equation}
    a_{j}^{k}(t)={p_{j}(t)}/{\delta_j}-\epsilon_{j,k}
\end{equation}
as the advantage score of $u_j$ on $B_k$. This quantity depends only on the intrinsic parameters of $u_j$: the transmit power $p_j(t)$ and the minimum-rate requirement via $\delta_j$, and the channel gain term $\epsilon_j^k$. A larger $a_j^k(t)$ implies a greater SINR surplus for $u_j$ and therefore makes \eqref{UA_cons} easier to satisfy for any fixed cochannel environment. In particular, $a_{j}^{k}(t)$ increases with higher power and better channel gain, and decreases with a stricter minimum-rate requirement, exactly matching the intuition of an advantage intrinsic to $u_j$.

The right-hand side (RHS) of \eqref{UA_cons}, $\sum_{j'\neq j} p_{j'}(t)$ aggregates the transmit powers of the other UEs sharing $B_k$. A larger RHS tightens the inequality and thus makes UE $u_j$'s minimum-rate constraint harder to satisfy. By symmetry, the transmit power of $u_j$ itself appears in the RHS of the constraints of all other UEs on $B_k$, i.e., $p_j(t)$ increases their cochannel load and degrades their feasibility. Hence $p_j(t)$ plays the role of an external burden on its neighbors and is naturally interpreted as a penalty term from the system perspective. Accordingly, we define the preference function as a weighted difference between these two components, i.e.,
\begin{equation}
PF_{o,i}^{\text{UA}}(j;t) = \varphi_{\text{UA}}\sum_{k\in \mathcal{B}_{o,i}}a_j^k(t)-\sum_{k\in \mathcal{B}_{o,i}}c_{k}\cdot p_{j}(t),
\end{equation}
\textcolor{black}{where $\varphi_{\mathrm{UA}}$ is a scaling coefficient that modulates the numerical contribution of the advantage score, and $c_{k}$ is the cost per unit of the interference caused by $u_j$ \cite{matching}. With this definition, a larger advantage score $a_{j,k}(t)$ increases the preference, whereas a higher cochannel load reduces it.}

The corresponding comparison function for satellites is obtained as
\begin{equation}
\begin{split}
(j,\mu_{\text{UA}})\succ_{o,i;t}^{\text{UA}}(j',\mu'_{\text{UA}})\Leftrightarrow PF_{o,i}^{\text{UA}}(j;t)\ge PF_{o,i}^{\text{UA}}(j';t).
\end{split}
\end{equation}
The profile list of $s_{o,i}$ sorting UEs is denoted as $\boldsymbol{L}^{\text{UA}}_{o,i}(t)$, and can be defined as
\begin{equation}
\boldsymbol{L}^{\text{UA}}_{o,i}(t)=\underset{{u_j\in\mathcal{U}_{o,i}}}{\operatorname{argsort}}\quad PF_{o,i}^{\text{UA}}(j).
\end{equation}
After executing \textcolor{black}{the UA module (Algorithm \ref{alg:UABA}, lines 1–19)}, we can have the set of UEs which change associated satellites as
\begin{equation}
\mathcal{U}^{\text{UA},[c_{i}]}_{o,i}=\big\{u_j\mid x_{o,i,j}^{[c_o+1,c_i]}(t)=1,x_{o,i,j}^{[c_o+1,c_i-1]}(t)=0\big\}.
\end{equation}
Through the definition of $\mathcal{U}^{\text{UA},[c_{i}]}_{o,i}$, we can evaluate the set of satellites which serve more UEs than the previous round as
\begin{equation}
\mathcal{S}^{\text{UA},[c_{i}]}=\{s_{o,i}\mid \mathcal{U}^{\text{UA},[c_{i}]}_{o,i}\neq \emptyset\}.
\end{equation}

\subsubsection{BA}    
\textcolor{black}{Because changing a UE's associated subcarrier will impact two subcarriers' sum rates, i.e., the previously associated subcarrier and the newly associated subcarrier. It is necessary to design the preference function considering impacts on these two subcarriers.} 
\textcolor{black}{
The preference function $PF1^{\text{BA}}(j; \boldsymbol{\mu})$ is introduced to describe the impact on the previously associated subcarrier and is calculated as 
\begin{equation}
    PF1^{\text{BA}}(j; \boldsymbol{\mu}) = R_{\mu_{\text{BA}}(j)}\left(\mu^{-1}_{\text{BA}}(\mu_{\text{BA}}(j))\setminus j\right)-R_{\mu_{\text{BA}}(j)},
\end{equation}
where $\mu^{-1}_{\text{BA}}(k)\setminus j$ denotes the set of UEs in $B_k$ excluding $u_j$. $R_{\mu_{\text{BA}}(j)}\left(\mu^{-1}_{\text{BA}}(\mu_{\text{BA}}(j))\setminus j\right)$ denotes the sum rate of UEs served by the subcarrier $\mu_{\text{BA}}(j)$ excluding $u_j$. It can be calculated through introducing the function $R_{k}\left(\mathcal{U}_x\right)$. The sum rate of UEs in ${\mathcal U}_x$ achieved on the $B_k$ is defined as
\begin{align}
R_{k}\left(\mathcal{U}_x\right)=&\sum\nolimits_{j\in\mathcal{U}_x}B  \notag\\
&\times \log_{2}\bigg[1+\frac{p_j(t)|H_{j}^k(t)|^2}{\displaystyle\sum\nolimits_{\substack{j' \ne j, j' \in \mathcal{U}_x}}p_{j'}(t)|H_{j}^k(t)|^2+\sigma^2}\bigg].
\end{align}
The preference function $PF2^{\text{BA}}(j,k; \boldsymbol{\mu})$ is introduced to describe the impact on the newly associated subcarrier and is calculated as 
\begin{equation}
PF2^{\text{BA}}(j,k; \boldsymbol{\mu}) = R_{k}\left(\mu^{-1}_{\text{BA}}(k)\cup j\right)-R_{k},
\end{equation}
where $\mu^{-1}_{\text{BA}}(k)\cup u_j$ denotes the set of UEs in $B_k$ plus $u_j$.
The profile function that is used to decide the profile list can be defined as the sum of $PF1^{\text{BA}}(j; \boldsymbol{\mu})$ and $PF2^{\text{BA}}(j,k; \boldsymbol{\mu})$
\begin{equation}\label{expr:BAprofileFuc}
PF^{\text{BA}}(j,k;\boldsymbol{\mu}) = PF1^{\text{BA}}(j; \boldsymbol{\mu}) + PF2^{\text{BA}}(j,k; \boldsymbol{\mu}).
\end{equation}}
The sorted profile list can then be constructed based on $PF(j,k;\boldsymbol{\mu})$, and  is denoted as $\boldsymbol{L}_{o,i}^{\text{BA}}(t)$. The BA profile list is constructed as
\begin{equation}\label{expr:BAprofile}
	\boldsymbol{L}_{o,i}^{\text{BA}}(t)=\underset{{u_j\in\mathcal{U}_{o,i},B_k\in\mathcal{B}_{o,i}}}{\operatorname{argsort}}\quad PF^{\text{BA}}(j,k).
\end{equation}
\textcolor{black}{Having detailed the UA and BA matching procedures, we analyze their fundamental properties. In particular, we prove that the matching results are stable and that the objective increases monotonically across iterations, ensuring convergence to a sub-optimal solution. It is formalized in Theorem \ref{theorem:matching}}.
\begin{theorem}\label{theorem:matching}
The matching $\boldsymbol{\mu}=\{{\mu_{\text{UA}},\mu_{\text{BA}}}\}$ is stable and can attain a sub-optimal solution given the transmitted power. 
\end{theorem}
\begin{proof}
    
    \textcolor{black}{According to the operation at the 14th line in the UA stage of Algorithm \ref{alg:UABA}, the matching $\mu_{\text{UA}}$ can make the quantity $\sum_j PF_{o,i}^{\text{UA}}(j)$ increase after each iteration, i.e., $\sum_j PF_{o,i}^{\text{UA}}(j)^{[r+1]}\ge\sum_j PF_{o,i}^{\text{UA}}(j)^{[r]}$.} It has been proved that $R(t)$ is non-decreasing with $\sum_j PF_{o,i}^{\text{UA}}(j)$\cite{matching}, then we have $R(\boldsymbol{x}^{[r+1]};\boldsymbol{y},\boldsymbol{z},\boldsymbol{p},t)\ge R(\boldsymbol{x}^{[r]};\boldsymbol{y},\boldsymbol{z},\boldsymbol{p},t)$. Clearly, $\mu_{\text{BA}}$ can also guarantee the increment of the network throughput given the UA result and transmitted power. Then we have $R(\boldsymbol{y}^{[r+1]};\boldsymbol{x},\boldsymbol{z},\boldsymbol{p},t)\ge R(\boldsymbol{y}^{[r]};\boldsymbol{x},\boldsymbol{z},\boldsymbol{p},t)$. With the UA matching results, we have $R(\boldsymbol{x}^{[c_{\text{o}}+1, c_{\text{i}}]}$,$\boldsymbol{y}^{[c_{\text{o}}+1, c_{\text{i}}]};\boldsymbol{p},t)\ge R(\boldsymbol{x}^{[c_{\text{o}}+1, c_{\text{i}}-1]}$,$\boldsymbol{y}^{[c_{\text{o}}+1, c_{\text{i}}-1]};\boldsymbol{p},t)$.
\end{proof}
	
\subsection{Lagrangian Dual-Based Optimization for PA}
After UA and BA stages, the remaining variables to be optimized are the UE transmit power $\boldsymbol{p}$. Moreover, for both the objective function and the constraints, the optimization problem \textbf{OP-0} can be reformulated as the form of the additive sum over satellites as
\begin{align}
	\textbf{OP-PA0}: &\max_{\boldsymbol{p}}  \quad \sum\nolimits_{(o,i)}\sum\nolimits_{j\in \mathcal{U}_{o,i}} \log_2{(1+\Gamma_{j}^{k})} \label{obj:rate2}\\
	\quad  \text {s.t. }& \sum\nolimits_{j\in \mathcal{U}_{o,i}}p_j(t) \leq P_{\text{max}}, \quad \forall o,i, \label{constraint_pmax2}\\
    & r_{j}(t) \geq r_{\text{min},j}(t), \quad \forall j\in \mathcal{U}_{o,i}, \forall o,i,\label{constraint_rmin2}\\
    & p_j(t)>0, \quad \forall j\in \mathcal{U}_{o,i}, \label{constraint_pj2}\forall o,i.
\end{align}
Because both the objective function (\ref{obj:rate2}) and the feasible set decided by constraints (\ref{constraint_pmax2})-(\ref{constraint_pj2}) decompose across satellites, the problem \textbf{OP-PA0} can be separated into $|S_{\boldsymbol{z}}|$ sub-problems corresponding to $|S_{\boldsymbol{z}}|$ satellites. The sub-problem corresponding to $s_{o,i}$ is denoted as $\textbf{OP-PA0}_{s_{o,i}}$ and formulated as
\begin{align}
	\textbf{OP-PA0}_{s_{o,i}}: &\max_{\boldsymbol{p}}  \quad \sum\nolimits_{j\in \mathcal{U}_{o,i}} \log_2{(1+\Gamma_{j}^{k})}\label{PAobj}\\
	\quad \text { s.t. }& \sum\nolimits_{j\in \mathcal{U}_{o,i}} p_j \leq P_{\text{max}}, \label{PAconstraint_pmax}\\
	& r_{j} \geq r_{\text{min},j}, \quad \forall j\in \mathcal{U}_{o,i},\label{PAconstraint_rmin}\\
	& p_j \ge 0, \quad \forall j\in \mathcal{U}_{o,i}.
\end{align}
Then we can further transform this problem through geometric programming. Let $\log_2{(1+\Gamma_{j}^{k}(t))}\simeq \log_2{(\Gamma_{j}^{k}(t))}$, $\acute{p}_j=\log_2 p_j$, and $\acute{I}_j^k=\log_2 I_j^k$, then $\textbf{OP-PA0}_{s_{o,i}}$ is transformed to $\textbf{OP-PA1}_{s_{o,i}}$ as
	\begin{align}
		\textbf{OP-PA1}_{s_{o,i}}: &\min_{\acute{\boldsymbol{p}},\acute{\boldsymbol{I}}}  \quad \sum\nolimits_{j\in \mathcal{U}_{o,i}} \log_2{\bigg[\frac{2^{-\acute{p}_j}}{|H_j^k|^2}(2^{\acute{I}_j^k}+\sigma^2)\bigg]}\label{PAobj2}\\
		\quad \text { s.t. }& \sum\nolimits_{j\in \mathcal{U}_{o,i}} 2^{\acute{p}_{j}} \leq P_{\text{max}}, \label{PAconstraint_pmax2}\\
		&\log_2{\bigg[\frac{2^{-\acute{p}_j}(2^{\acute{I}_j^k}+\sigma^2)}{|H_{j}^{k}(t)|^2}\bigg]}+\chi_j\leq 0, \quad \forall j\in \mathcal{U}_{o,i},\label{PAconstraint_rmin2}\\
		& -\acute{I}_j^k+\log_2 I_j^k = 0, \quad \forall j\in \mathcal{U}_{o,i}\label{Cons_pI},
	\end{align}
	where $\chi_j=\frac{r_{\text{min},j}}{B}$, $\forall j\in \mathcal{U}_{o,i}$. 
	\textcolor{black}{By invoking the KKT framework and introducing the multipliers $\lambda_{o,i}$, $\boldsymbol{n}$, and $\boldsymbol{m}$ for the constraints (\ref{constraint_pmax2}), (\ref{constraint_rmin2}), and (\ref{constraint_pj2}), respectively, $\textbf{OP-PA1}_{s_{o,i}}$ can be transformed to Lagrangian saddle-point problem $\textbf{OP-PA2}_{s_{o,i}}$, i.e.,}
	\begin{align}
		\textbf{OP-PA2}_{s_{o,i}}: &\min_{\acute{\boldsymbol{p}},\acute{\boldsymbol{I}}}\max_{\lambda_{o,i},\boldsymbol{n},\boldsymbol{m}}\quad L_{o,i}\\
		\quad \text{ s.t. } & \lambda_{o,i}\ge 0,\boldsymbol{n}\ge \boldsymbol{0}, \boldsymbol{m}\ge \boldsymbol{0},
	\end{align}
	where Lagrange term $L_{o,i}$ can be calculated as
    \begin{equation}
		\begin{aligned}
			L_{o,i}=&\underset{j\in \mathcal{U}_{o,i}}{\sum} \log_2{\bigg[\frac{2^{-\acute{p}_j}}{|H_{j}^k|^2}(2^{\acute{I}_j^k}+\sigma^2)\bigg]}\\
			&+\lambda_{o,i}(\sum\nolimits_{j\in \mathcal{U}_{o,i}} 2^{\acute{p}_j}-P_{\max})\\
			&+\underset{j\in \mathcal{U}_{o,i}}{\sum} n_j\bigg(\log_2\bigg[\frac{2^{-\acute{p}_j}}{|H_{j}^k|^2}(2^{\acute{I}_j^k}+\sigma^2)\bigg]+\chi_j\bigg)\\
			&+\sum\nolimits_{j\in \mathcal{U}_{o,i}} m_j(-\acute{I}_j^k+\log_2 I_j^k).
		\end{aligned}
	\end{equation}
	
To solve $\textbf{OP-PA2}_{s_{o,i}}$, we apply the duality function evaluated as $D_{o,i}(\lambda_{o,i},\boldsymbol{n},\boldsymbol{m})=\min_{\acute{\boldsymbol{p}},\acute{\boldsymbol{I}}} L_{o,i}.$
	
Then the duality problem can be defined as 
\begin{align}
	\textbf{OP-PA3}_{s_{o,i}}: & \max_{\lambda_{o,i},\boldsymbol{n},\boldsymbol{m}}\quad D_{o,i}(\lambda_{o,i},\boldsymbol{n},\boldsymbol{m})\\
	\quad \text { s.t. }& \lambda_{o,i}\ge 0,\boldsymbol{n}\ge \boldsymbol{0}, \boldsymbol{m}\ge \boldsymbol{0}.
\end{align}
Since the strong duality holds here \cite{matching}, $\textbf{OP-PA3}_{s_{o,i}}$ is equivalent to $\textbf{OP-PA2}_{s_{o,i}}$ and the optimal primal and dual variables $\acute{\boldsymbol{p}}$ and $\acute{\boldsymbol{I}}$ are characterized by the saddle point of the Lagrangian. Fixing the multipliers $\lambda_{o,i}$, $\boldsymbol{n}$, and $\boldsymbol{m}$, we apply the KKT stationary conditions by differentiating $L_{o,i}$ w.r.t. $\acute{\boldsymbol{p}}$ and $\acute{\boldsymbol{I}}$ and setting the results to zero as
    
	\begin{align}
		{\partial L_{o,i}}/{\partial \acute{p}_j}&=-(1+n_j)+\lambda_{o,i} 2^{\acute{p}_j}\cdot \ln2=0,\\
		{\partial L_{o,i}}/{\partial \acute{I}^k_j}&=(1+n_j){2^{\acute{I}_j^k}}/({2^{\acute{I}_j^k}+\sigma^2})-m_j=0.
	\end{align}
	Then the optimal $p_j$ and $I_j^k$ are evaluated as
	\begin{align}
		p_j^{*}&=2^{\acute{p}_j^{*}}={(1+n_j)}/{(\lambda_{o,i}\ln2)},\label{pj}\\
		I_j^{k*}&=2^{\acute{I}_j^{k*}}=\max\{{m_j\sigma^2}/{(1+n_j-m_j)},  0\}.\label{Ij}
	\end{align}
	
	\begin{algorithm}[tb]
		\caption{Algorithm of PA}\label{alg:PA}
		\begin{algorithmic}[1]
			\ENSURE{$\boldsymbol{p}^{[c_{\text{o}}+1, c_{\text{i}}]}$}
			\REQUIRE{$\boldsymbol{x}^{[c_{\text{o}}+1, c_{\text{i}}]},\boldsymbol{y}^{[c_{\text{o}}+1, c_{\text{i}}]},\boldsymbol{p}^{[c_{\text{o}}+1, c_{\text{i}}-1]}$}
			\STATE \textbf{Initialization:} $\boldsymbol{\lambda}^{(0)},\boldsymbol{n}^{(0)},\boldsymbol{m}^{(0)},\boldsymbol{p}^{(0)},\boldsymbol{I}^{(0)}$
			\FOR {each $s_{o,i}\in \mathcal{S}_{\boldsymbol{z}}$}
			\STATE $r=0$
			\WHILE {$\acute{\boldsymbol{p}}$ does not converge}
			\IF {r == 0}
			\STATE Update $\lambda_{o,i}^{(1)}$, $n_j^{(1)}$ and $m_j^{(1)}$ according to (\ref{lamdaj}), (\ref{nj}) and (\ref{mj}) based on the values of $\boldsymbol{p}^{(0)},\boldsymbol{I}^{(0)}$
			\ELSE
			\STATE Update $\acute{p}_j^{(r)*}$ and $\acute{I}_j^{k,(r)*}$ based on (\ref{pj}) and (\ref{Ij})
			\STATE Calculate the value of $I_j^{k,(r)}$ based on the value of $\acute{\boldsymbol{p}}^{(r)*}$ according to (\ref{expr:Interfernce})
			\STATE Update $\lambda_{o,i}^{(r+1)}$, $n_j^{(r+1)}$ and $m_j^{(r+1)}$ according to (\ref{lamdaj}), (\ref{nj}) and (\ref{mj}) based on the values of $\acute{\boldsymbol{p}}^{(r)*}$, $I_j^{k,(r)}$ and $\acute{I}_j^{k,(r)*}$
			\ENDIF
			\STATE $r=r+1$
			\ENDWHILE
			\ENDFOR
		\end{algorithmic}
	\end{algorithm}

Based on the optimal $p_j^{*}$ and $I_j^{k*}$, Lagrange multipliers $\lambda_{o,i}$, $n_j$ and $m_j$ can be updated recursively with step sizes $q_{\lambda}^{(r)}$, $ q_{n}^{(r)}$ and $ q_{m}^{(r)}$ according to the subgradient method as
	\begin{align}
		\lambda_{o,i}^{(r)}&=\lambda_{o,i}^{(r-1)} + q_{\lambda}^{(r)}\cdot g_{\lambda_{o,i}}^{(r-1)}\label{lamdaj},\\
		n_j^{(r)}&=n_j^{(r-1)} + q_{n_j}^{(r)}\cdot g_{n_j}^{(r-1)}\label{nj},\\
		m_j^{(r)}&=m_j^{(r-1)} + q_{m_j}^{(r)}\cdot g_{m_j}^{(r-1)}\label{mj},
	\end{align}
	where the gradients are calculated as
	\begin{align}
		g_{\lambda_{o,i}}^{(r-1)}&=\sum\nolimits_{j\in \mathcal{U}_{o,i}} 2^{\acute{p}^{(r-1)*}_j}-P_{\max},\\
		g_{n_j}^{(r-1)}&=\log_2\bigg[\frac{2^{-\acute{p}^{(r-1)*}_j}}{|H_j^k|^2}(2^{\acute{I}^{k,(r-1)*}_j}+\sigma^2)\bigg]+\chi_j\label{vnj},\\
		g_{m_j}^{(r-1)}&=\log_{2}\bigg[I_j^{k,(r-1)}(\acute{\boldsymbol{p}}^{(r-1)*})\bigg]-\acute{I}_j^{k,(r-1)*},\label{vmj}
	\end{align}
    where in (\ref{vmj}), $I_j^{k,(r-1)}(\acute{\boldsymbol{p}}^{(r-1)*})$ is evaluated via (\ref{expr:Interfernce}) after obtaining the transmit power $\acute{\boldsymbol{p}}^{(r-1)*}$, i.e.,
    \begin{equation}
        I_j^{k,(r-1)}=I_j^{k,(r-1)}(\acute{\boldsymbol{p}}^{(r-1)*})=\sum_{j'\ne j}p_{j'}^{(r-1)*}(t)|H_{j}^k(t)|^2.
    \end{equation}
	Accordingly, we explicitly write $I_j^{k,(r-1)}$ as $I_j^{k,(r-1)}(\acute{\boldsymbol{p}}^{(r-1)*})$ to emphasize its functional dependence on the power vector. 
    
    Since every part of JUBPA has been illustrated, we will analyze the convergence of JUBPA by proposing the following Theorem 2.
\begin{theorem}\label{theorem:converge}
    Algorithm \ref{alg:JUBPA} can converge and achieve a sub-optimal solution.
\end{theorem}
\begin{proof}
 To show the convergence of Algorithm \ref{alg:JUBPA}, we need to prove the monotonicity of network rate function $R(\cdot)$ along iteration steps, which is the function of $\boldsymbol{x}$, $\boldsymbol{y}$, and $\boldsymbol{z}$. The first inequality corresponding to PA stage is
 \begin{align}
     &R(\poic)|_{[\xoic,\yoic]} \ge \nonumber \\
     &R(\poici)|_{[\xoic,\yoic]},
 \end{align}
where $R(\bigstar)|_{[\blacktriangle]}$ represents the network rate function w.r.t $\bigstar$ under fixed parameters $\blacktriangle$. The inequality holds because the PA part is designed based on the gradient type of method. Thus the non-decreasing monotonicity can be achieved with respect to $\boldsymbol{p}$. The second inequality corresponding to UA and BA stages is expressed as 
\begin{align}
     &R(\xoic,\yoic)|_{\left\{\poici\right\}} \ge \nonumber\\
     &R(\xoici,\yoici)|_{\left\{\poici\right\}},
 \end{align}

which has been proved in Theorem \ref{theorem:matching}.
\end{proof}

\subsection{Complexity Analysis}
As illustrated in Algorithm \ref{alg:markovapproximation}, the proposed \methodname{} algorithm is composed of an inner \stagename{} routine and an outer Markov approximation procedure.
Given a selected satellite subset, \stagename{} alternates between UA\&BA and PA. In UA/BA, UEs propose the requirements of changing associated satellites and subcarriers following preference functions until a stable matching is reached. UA\&BA has the complexity of ${\mathcal O}(Z_{\text{th}}J+JK/Z_{\text{th}})$. 
Subsequent PA solves the power allocation subproblem via dual decomposition. In PA, each satellite updates its corresponding multipliers and power variables. This leads to the complexity of ${\mathcal O}(R_{\text{PA}}Z_{\text{th}})$ where $R_{\text{PA}}$ is the number of iterations. Therefore, \stagename{} has the complexity of ${\mathcal O}(Z_{\text{th}}J+JK/Z_{\text{th}}+R_{\text{PA}}Z_{\text{th}})$, and Algorithm  \ref{alg:markovapproximation} will invoke JUBPA $C_0$  times with the complexity ${\mathcal O}(C_0[Z_{\text{th}}J+JK/Z_{\text{th}}+R_{\text{PA}}Z_{\text{th}}])$.
\textcolor{black}{In each outer iteration, we have the current satellite subset $\mathcal{S}1$ with its objective value $R1$ computed by the inner JUBPA. A candidate subset $\mathcal{S}2$ is generated by exploration, and its objective $R2$ is obtained by one additional call to JUBPA. The Markov update then only computes the performance difference $\Delta R = R2 - R1$, evaluates the transition probability via (\ref{eq:prob_trans1}), draws a uniform random variable, and selects either $\mathcal{S}1$ or $\mathcal{S}2$ accordingly. The complexity of these operations is $\mathcal{O}(1)$. Thus the per-iteration cost is dominated by the inner JUBPA evaluation.}

\section{Simulation Results}\label{Simulation Results}
	\begin{figure}[t]
		\centering
		\subfigure[]{
			\includegraphics[width=4.2cm,height=3.7cm]{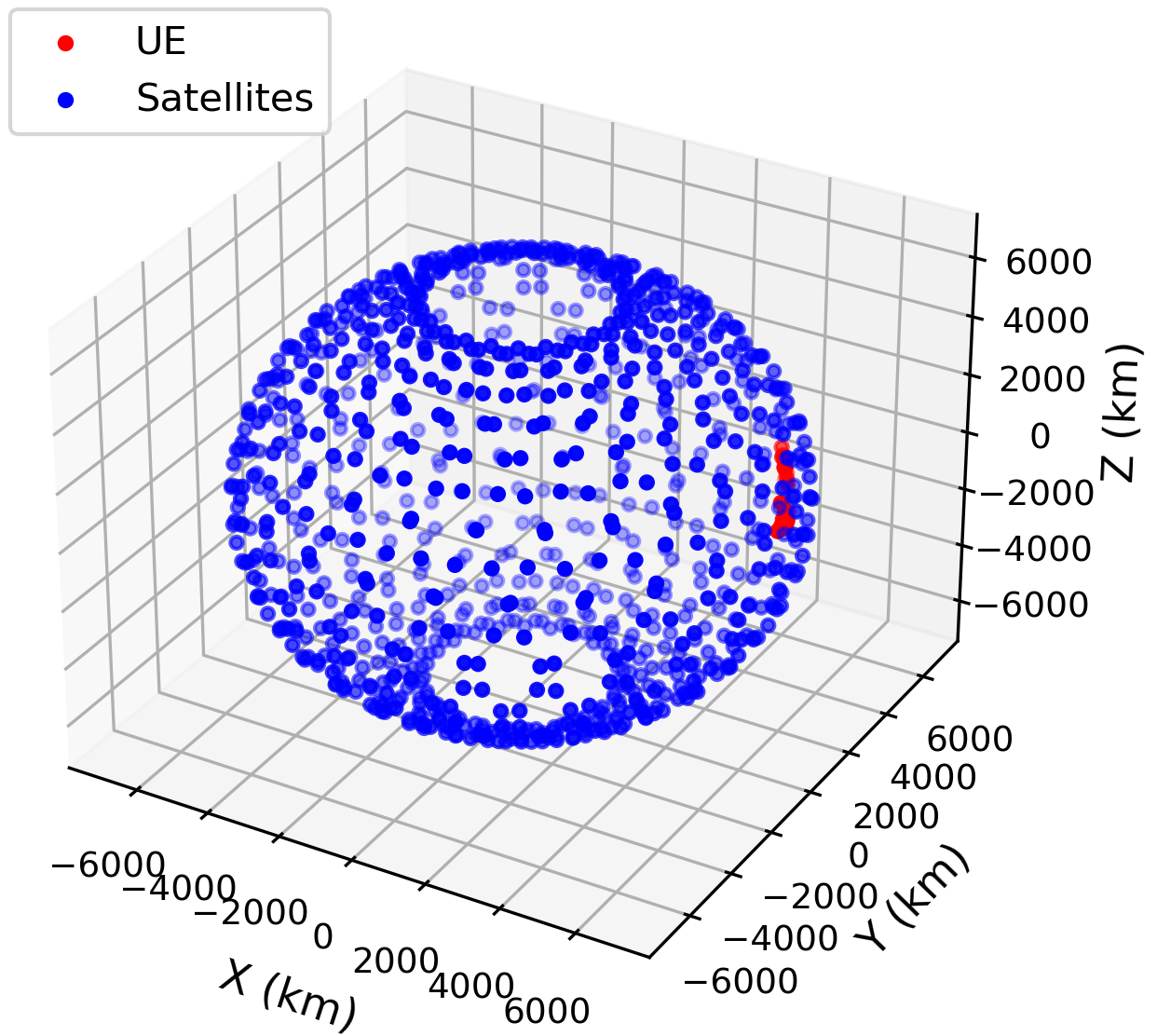}
		}
\hspace{-1em}
		\subfigure[]{
			\includegraphics[width=4.2cm,height=3.7cm]{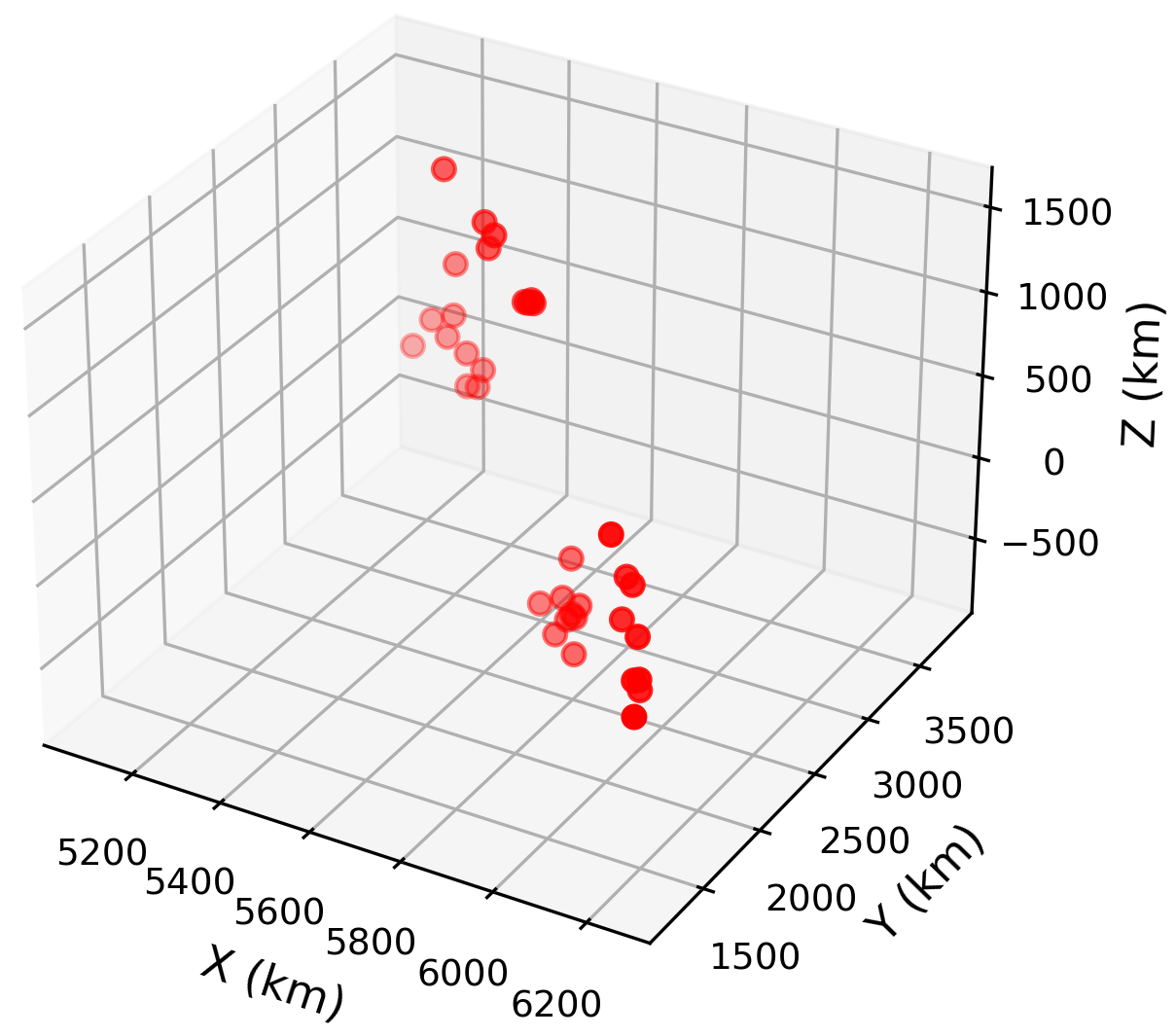}
		}
		\caption{Spatial distributions of (a) satellites with UEs and (b) UEs.}
		\label{fig:sat_coordinates}
	\end{figure}

	\begin{table}[t]
		\centering
		
		\caption{General System Parameters}\label{tab:simulationpara}
		\begin{tabular}{|p{0.2\textwidth}|c|} 
			\hline
			Parameters & Description\\
			\hline
			Number of satellites, $\abs{\mathcal{S}}$ & $25\times 40$ \\
			\hline
			Number of UEs, $\abs{\mathcal{U}}$  & $30$\\
			\hline
			Number of subcarriers, $\abs{\mathcal{B}}$  & $25$\\
			\hline
			Central frequency, $f_c$  & $6$ GHz\\
			\hline
			Bandwidth, $B$  & $10$ MHz\\
			\hline
			Orbits height, $h_{\text{sat}}$  & $550$ km\\
			\hline
			Orbits inclination  & $70\degree$\\
			\hline
			Maximum power, $P_{\text{max}}$ & 5 W\\
			\hline
			Satellites amount, $Z_{\text{th}}$ & 10\\
			\hline
			Shadowing factor, $SF$ & 1 dB \\
			\hline
			Antenna gain, $G_{t}$ & 30 dB\\
			\hline
			Cone angle, $\varphi$ (rad) & $5\pi/12$ \\
			\hline
			UEs area center, $\boldsymbol{a}_{\text{UEC}}$ (km) & $[5765.39, 2717.12, 252.39]$ \\
			\hline
		\end{tabular}
	\end{table}

\subsection{\leoname{} System Settings}
To evaluate \methodname{} in a more practical setting, we consider a {\em Walker Constellation Design Pattern } for an \leoname{} satellite communication system.
The constellation consists of $1{,}000$ satellites deployed at different orbits with the same altitude of $550\mbox{ km}$ and a common inclination of $70\degree$ but different right ascension of the ascending nodes.
The communication system operates at a carrier frequency of $f_c=6\mbox{ GHz}$, and each subcarrier occupies a bandwidth of $10\mbox{ MHz}$. 
$25$ UEs are uniformly distributed within a circular area of radius $1{,}500\mbox{ km}$, centered at $\boldsymbol{a}_{\text{UEC}}$.
The coordinates of the satellites constellation and UEs are shown in Fig.~\ref{fig:sat_coordinates}. 
Throughout the simulations, unless otherwise specified, we adopt the parameters reported in Table~\ref{tab:simulationpara} (see \cite{nasa_smallsat_link_2019} and references therein).

\subsection{\methodname{} Simulations}
\subsubsection{\methodname{} Convergence}
In Fig.\ref{fig:Markov_converge}, we can see that \methodname{}, implemented via Algorithm \ref{alg:markovapproximation}, exhibits a consistent convergence behavior under different shadowing conditions ($SF = 1, 2, \text{ and } 3\text{~dB}$).
As expected, better channel quality (smaller shadowing factor) leads to higher achievable sum rates. 
Under severe fading scenarios (e.g., $SF=3$ dB), it becomes more challenging to identify satellite groups that jointly satisfy the system constraints and yield high data rates. Consequently, the algorithm stabilizes more quickly, leading to an earlier decay of the exploration probability. Conversely, under milder fading conditions (e.g., $SF=1$ dB), \methodname{} continues exploring larger candidate sets, resulting in better final performance but requiring more iterations to converge.
Occasional drops may still occur even after the sum rate reaches high values. \textcolor{black}{This behavior arises from the Markov approximation. It happens because the algorithm continuously explores alternative satellite constellations to achieve marginal performance gains even after achieving a high objective function value. As a result, satellite groups corresponding to lower function values can still be found and cause fluctuation.} Fig.~\ref{fig:Markov_CDF} further illustrates this behavior by showing the cumulative distribution function (CDF) of the achieved sum rates. The CDF curves show that smaller shadowing factors yield a clear rightward shift, indicating consistently higher rates across UEs. This validates not only the convergence of \methodname{} but also its robustness in adapting to different channel conditions.
\begin{figure}[!t]
	\centerline{\includegraphics[width=0.4\textwidth]{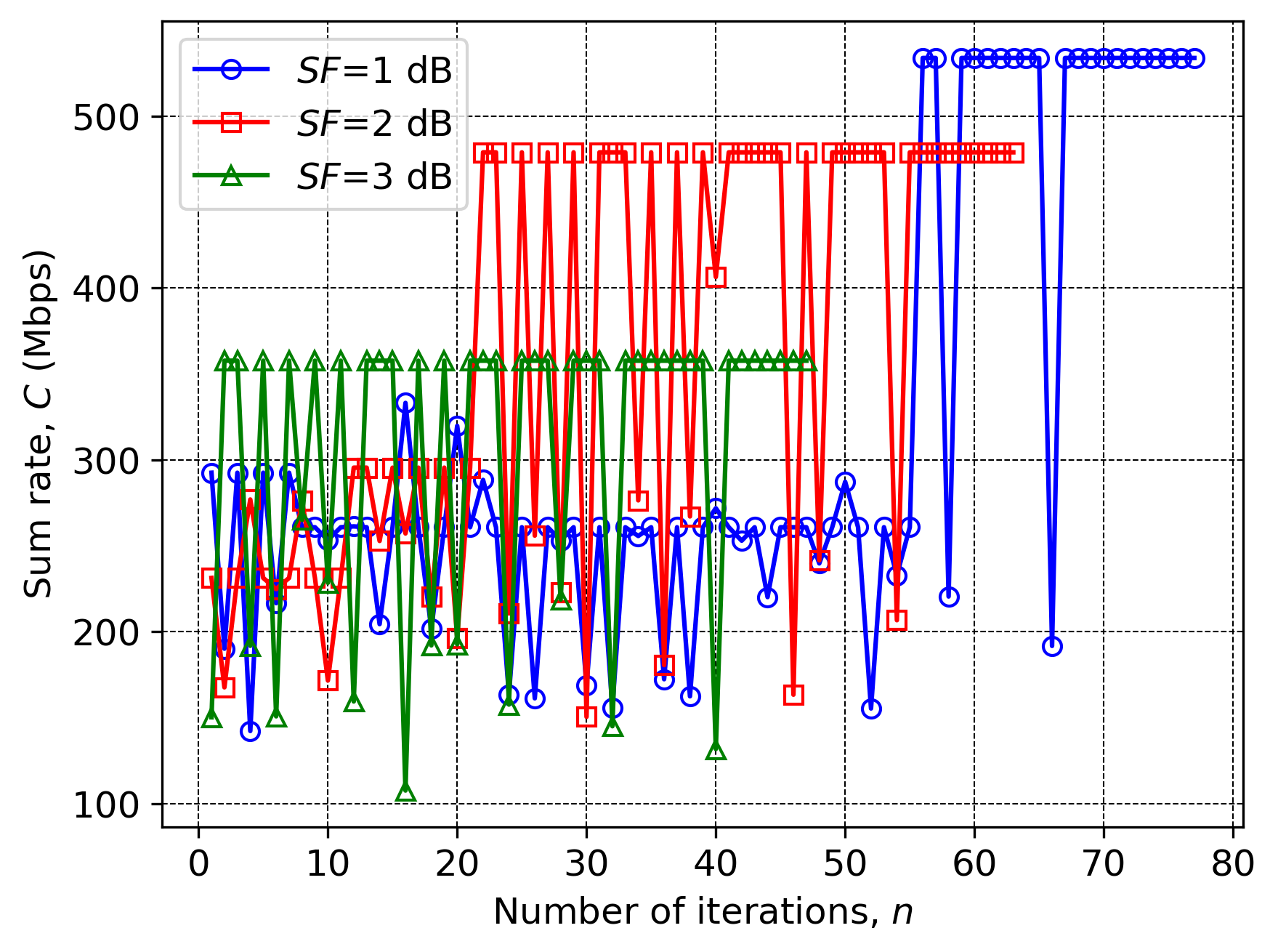}}
	\caption{Convergence of \methodname{} with $P_{\max}=50\mbox{ W}$.}
	\label{fig:Markov_converge}
\end{figure}
\begin{figure}[!t]
	\centerline{\includegraphics[width=0.4\textwidth]{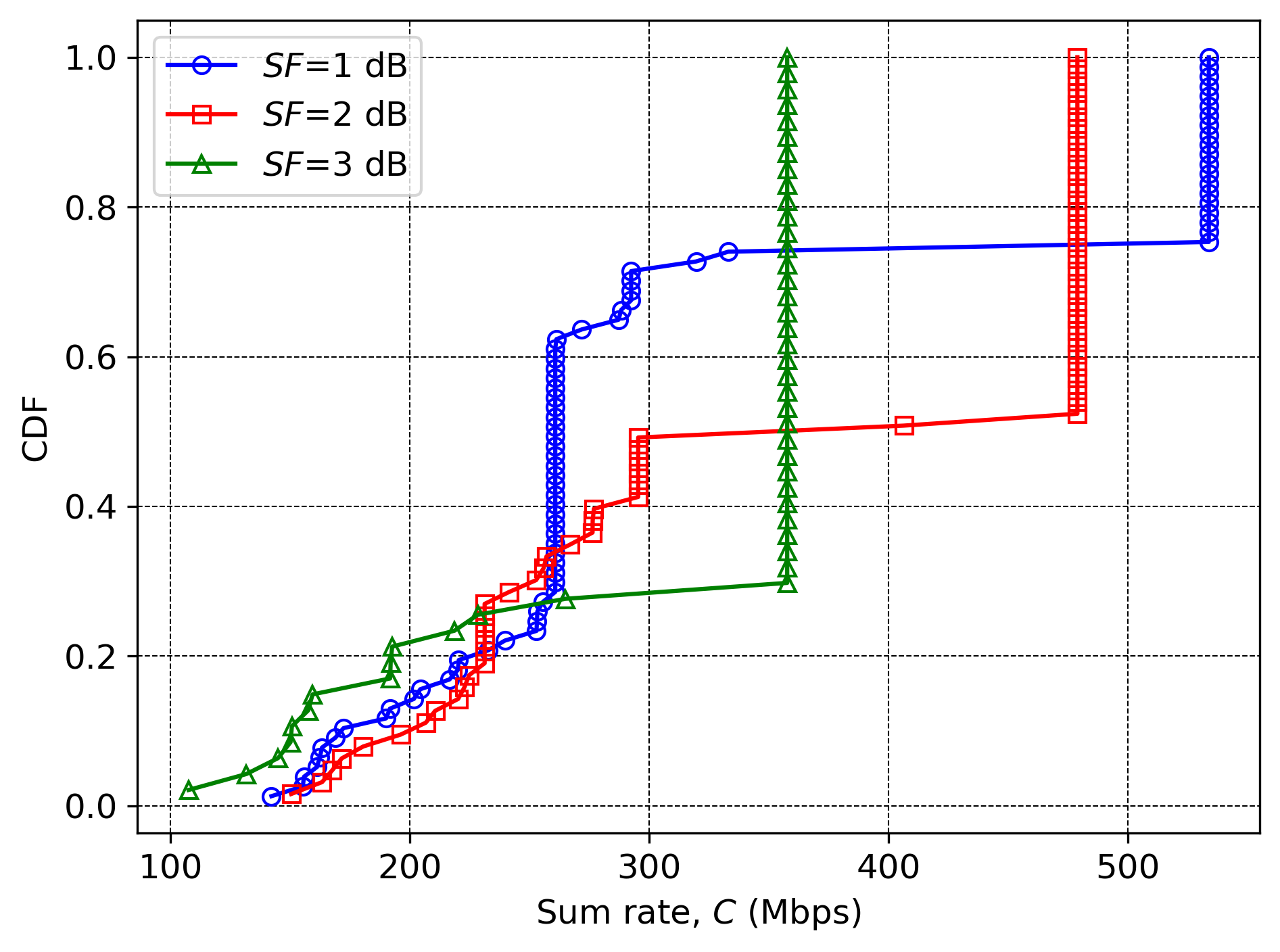}}
	\caption{\methodname{}'s CDF curve with $P_{\max}=50\mbox{ W}$.}
	\label{fig:Markov_CDF}
\end{figure}
\subsubsection{Performance of Matching-game based UA and BA in \methodname{}}
Fig.~\ref{fig:UABA_converge} illustrates the convergence behavior of the matching-based UA and BA (Algorithm \ref{alg:UABA}) procedures with a fixed PA policy. The power of each UE is obtained by the initialization method of JUBPA.
The results show that the UA-BA combination can converge reliably under different network sizes ($|\mathcal{U}|=20$, $25$, and $30$). 
Notably, the sum rate improves significantly after the first iteration, but subsequent iterations yield limited gains. 
This is because, without invoking Algorithm \ref{alg:PA}, interference cannot be effectively managed. This limits further performance improvement. 
Moreover, the sum rate grows with the number of UEs, since a larger user pool increases the likelihood of favorable satellite-UE matching and more efficient bandwidth utilization.

\begin{figure}[tb]
	\centerline{\includegraphics[width=0.4\textwidth]{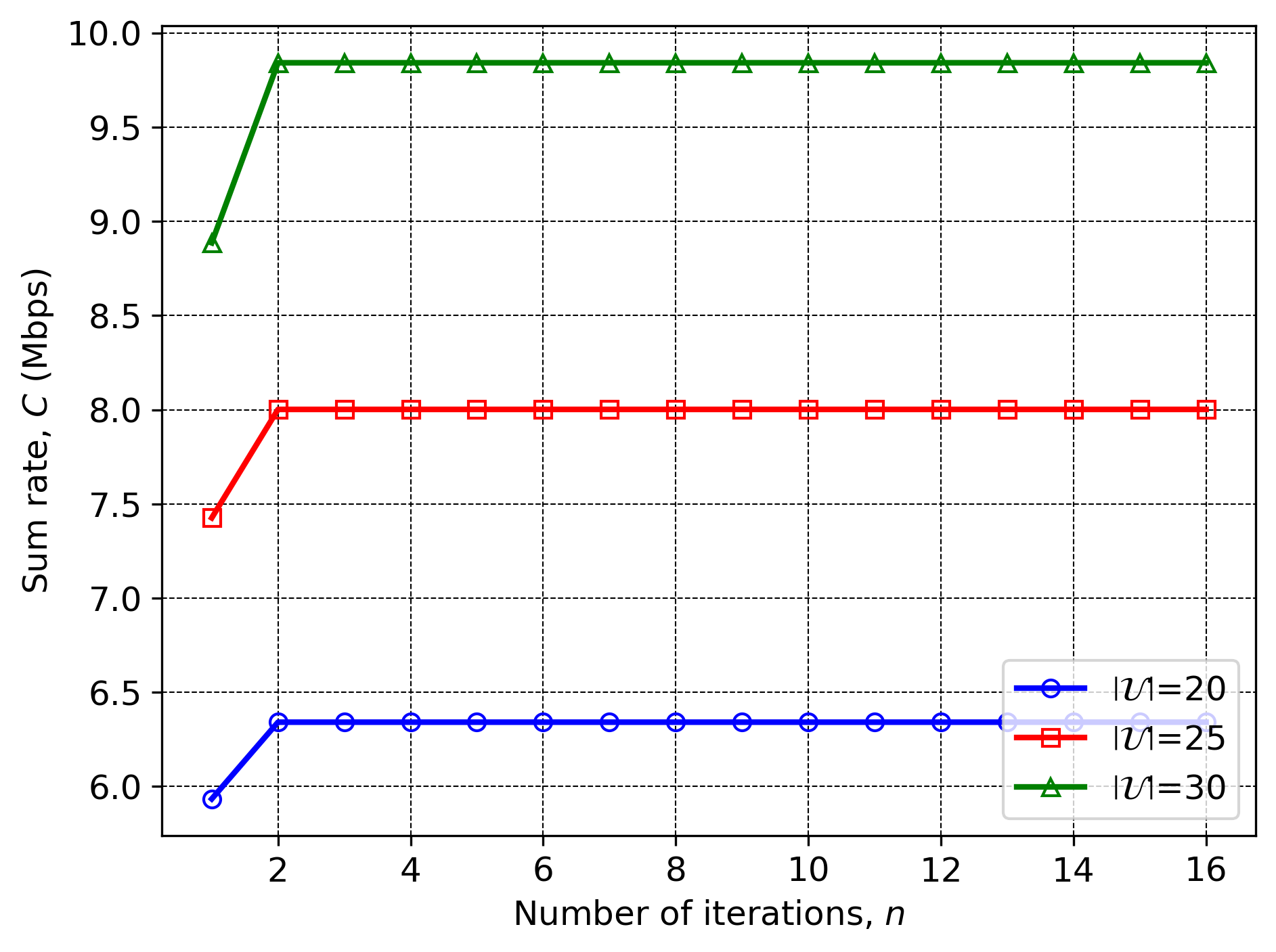}}
	\caption{Convergence of matching process ($r_{\text{min}}=0.3$ Mbps).}
	\label{fig:UABA_converge}
\end{figure}
\begin{figure}[tb]
	\centerline{\includegraphics[width=0.4\textwidth]{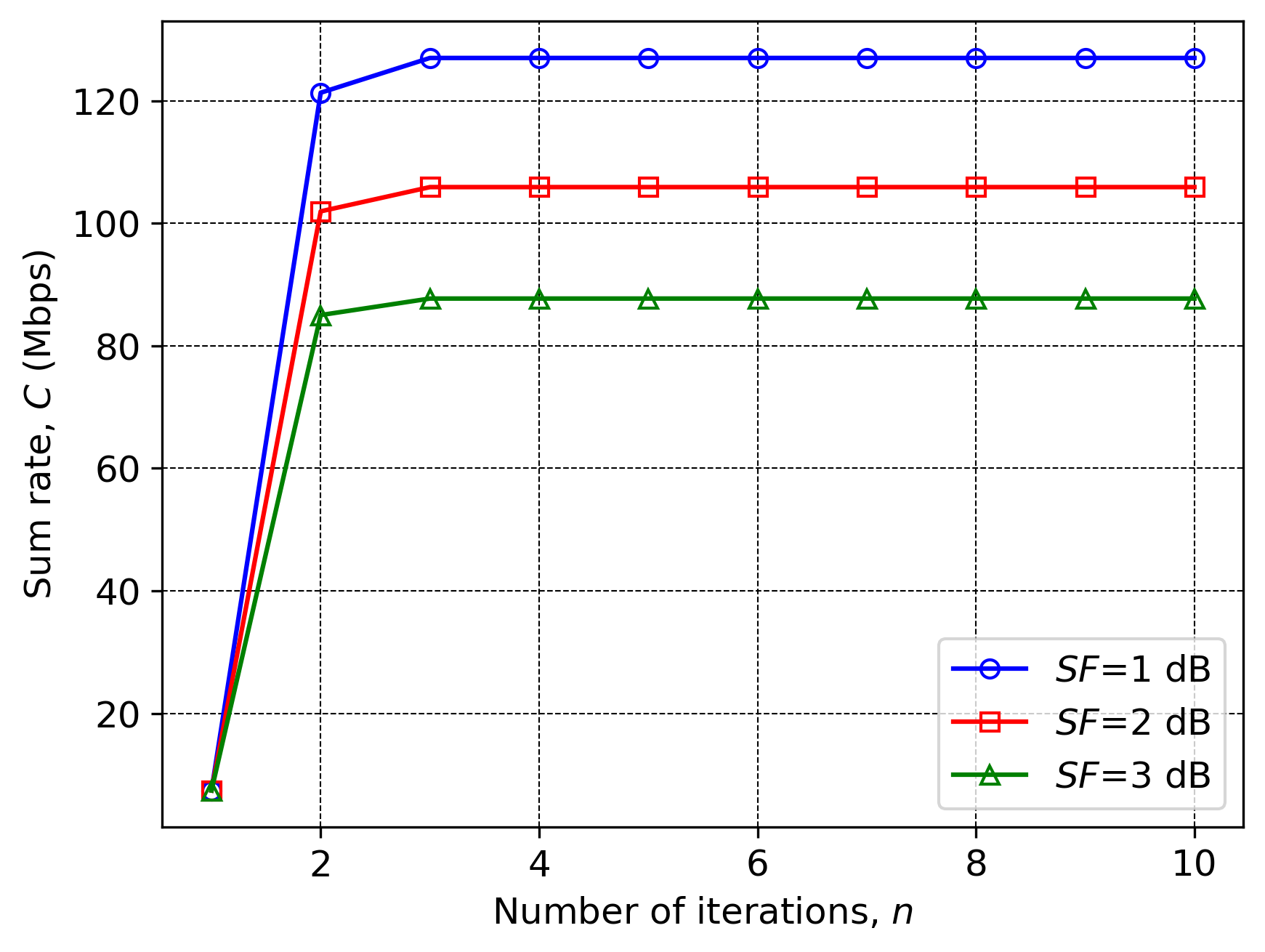}}
	\caption{Convergence of \stagename{} algorithm.}
	\label{fig:URP_converge}
\end{figure}
\subsubsection{Effectiveness of \stagename{} in \methodname{}} 
Fig.~\ref{fig:URP_converge} shows the convergence behavior of the proposed \stagename{} algorithm under different shadowing factors $SF=1,2,3$ dB. 
The results show that \stagename{} achieves stable convergence across all channel conditions. 
The results are consistent with the trend shown in Fig.~\ref{fig:Markov_converge}
.
As expected, better channel quality, corresponding to smaller $SF$, leads to higher achievable sum rates, since weaker shadowing implies less blockage and more reliable links. 
Moreover, \stagename{} converges rapidly, typically within two to three iterations. This highlights its efficiency in practical implementations.
Compared with the UA-BA combination, \stagename{} yields a pronounced performance gain, underscoring the effectiveness of jointly incorporating PA into the matching-based UA and BA.
 
\subsection{Performance Comparison}
\subsubsection{Gains of Joint UA-BA Matching in \methodname{}}

In order to verify the effectiveness of the matching game-based UA and BA methods, we evaluate the performance of \methodname{} against the following two methods: (i) Fixed-UA + Fixed-BA: both UA and BA are predetermined in a static manner without applying matching game theory;
(ii) Fixed-UA + Matching-BA: the UA is fixed, and the BA is optimized via the matching game algorithm (Algorithm \ref{alg:UABA}).
These methods are testified under the same Markov approximation and PA procedures, ensuring a fair comparison separating the contribution of matching game-based UA and BA.
The results are presented in Fig.~\ref{fig:UB_matching_comparison}, where the sum rate performance is plotted as a function of the cone angle $\varphi$ varying from $\frac{\pi}{10}$ to $\frac{7\pi}{20}$ in a step size of $\frac{\pi}{40}$. 
The achievable sum rate generally improves with larger $\varphi$ since a wider cone angle allows more satellites to be considered in the association and beam assignment process. 
The performance gaps among three methods highlight the importance of applying matching game theory to both UA and BA. 
Specifically, the double matching scheme consistently outperforms the baselines, demonstrating the synergistic gain of jointly adapting both UA and PA. 
In contrast, Fixed-UA + Matching-BA method achieves moderate improvement over the fixed baseline, but its performance saturates when $\varphi$ grows as fixing UA limits the ability to exploit the enlarged candidate satellite set. 
Finally, the fixed method exhibits the poorest performance and even suffers performance degradation at higher $\varphi$, where interference becomes more significant and cannot be effectively mitigated without adaptive UA and BA.
\begin{figure}[!t]
		\centerline{\includegraphics[width=0.4\textwidth]{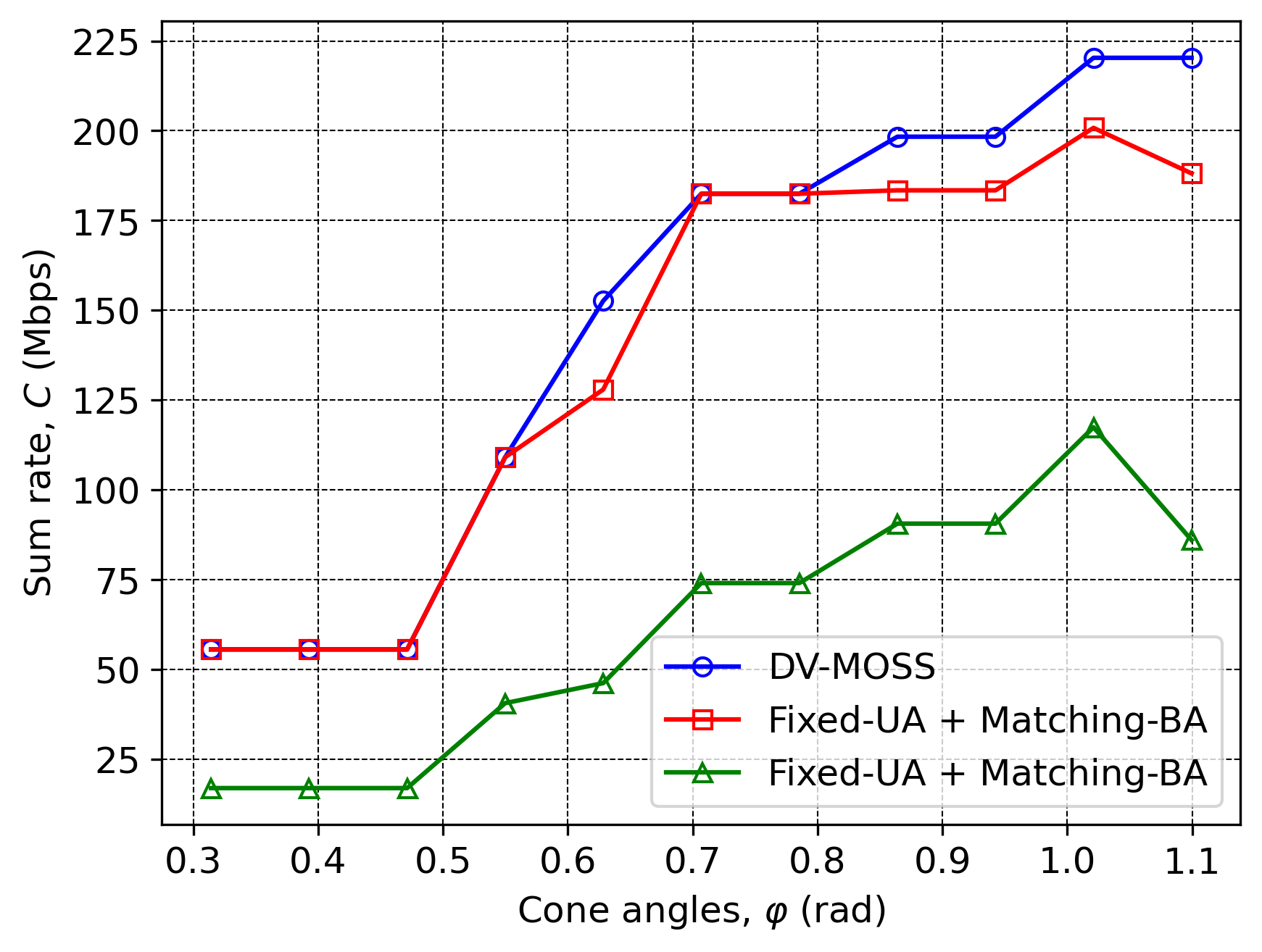}}
		\caption{Comparison among matching algorithms with $\abs{\mathcal{U}}=35, \varphi \in \left[\frac{\pi}{10},\frac{7\pi}{20}\right]$.}
		\label{fig:UB_matching_comparison}
	\end{figure}
	
\begin{figure}[!t]        \centerline{\includegraphics[width=0.4\textwidth]{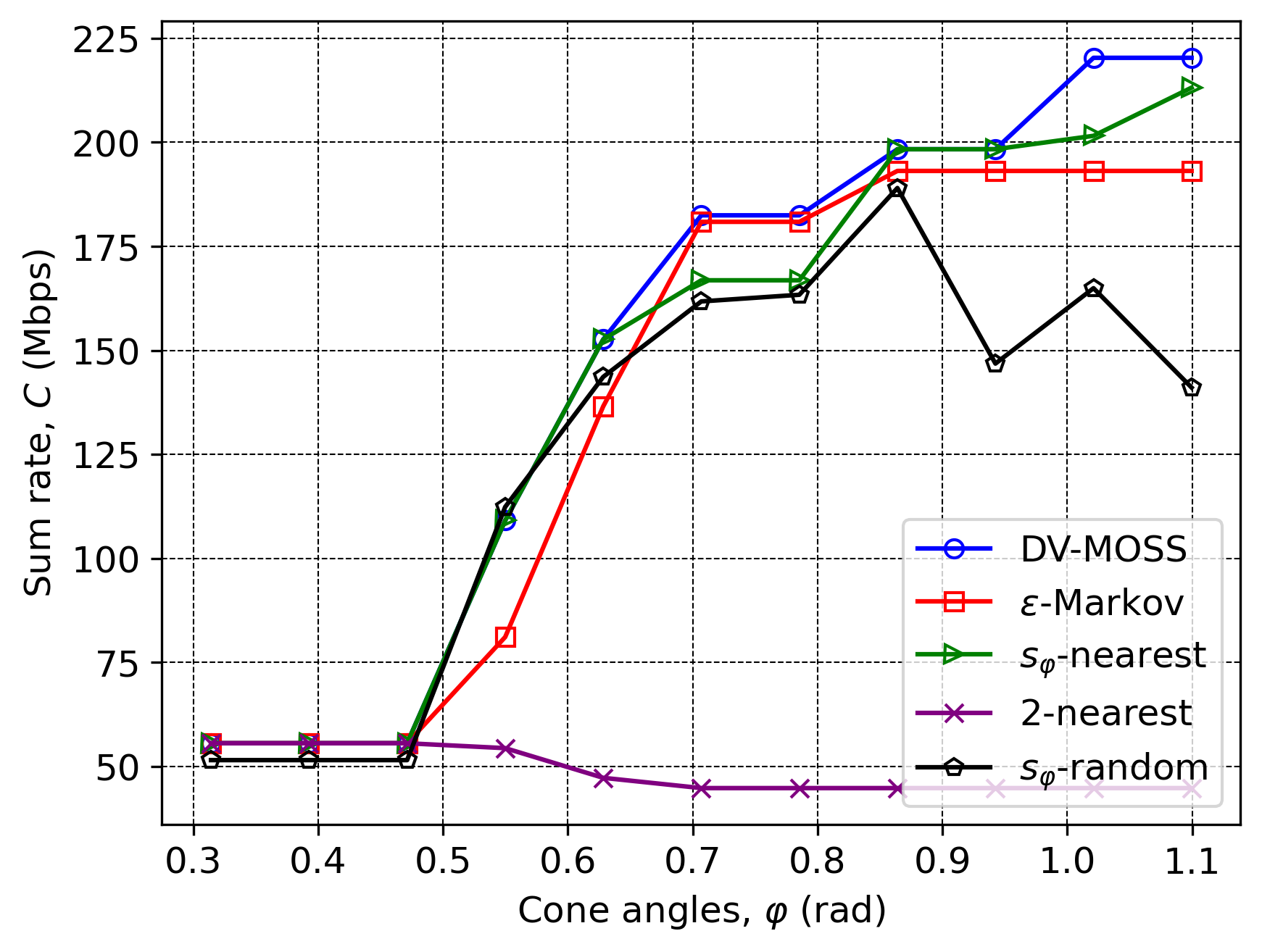}}
\caption{Comparison among selection algorithms with $\abs{\mathcal{U}}=35, \varphi \in \left[\frac{\pi}{10},\frac{7\pi}{20}\right]$.}
\label{fig:PerformanceBenchmark}
\end{figure}

\begin{figure*}[t]
		\centering
		\subfigure[]{\includegraphics[width=0.22\textwidth]{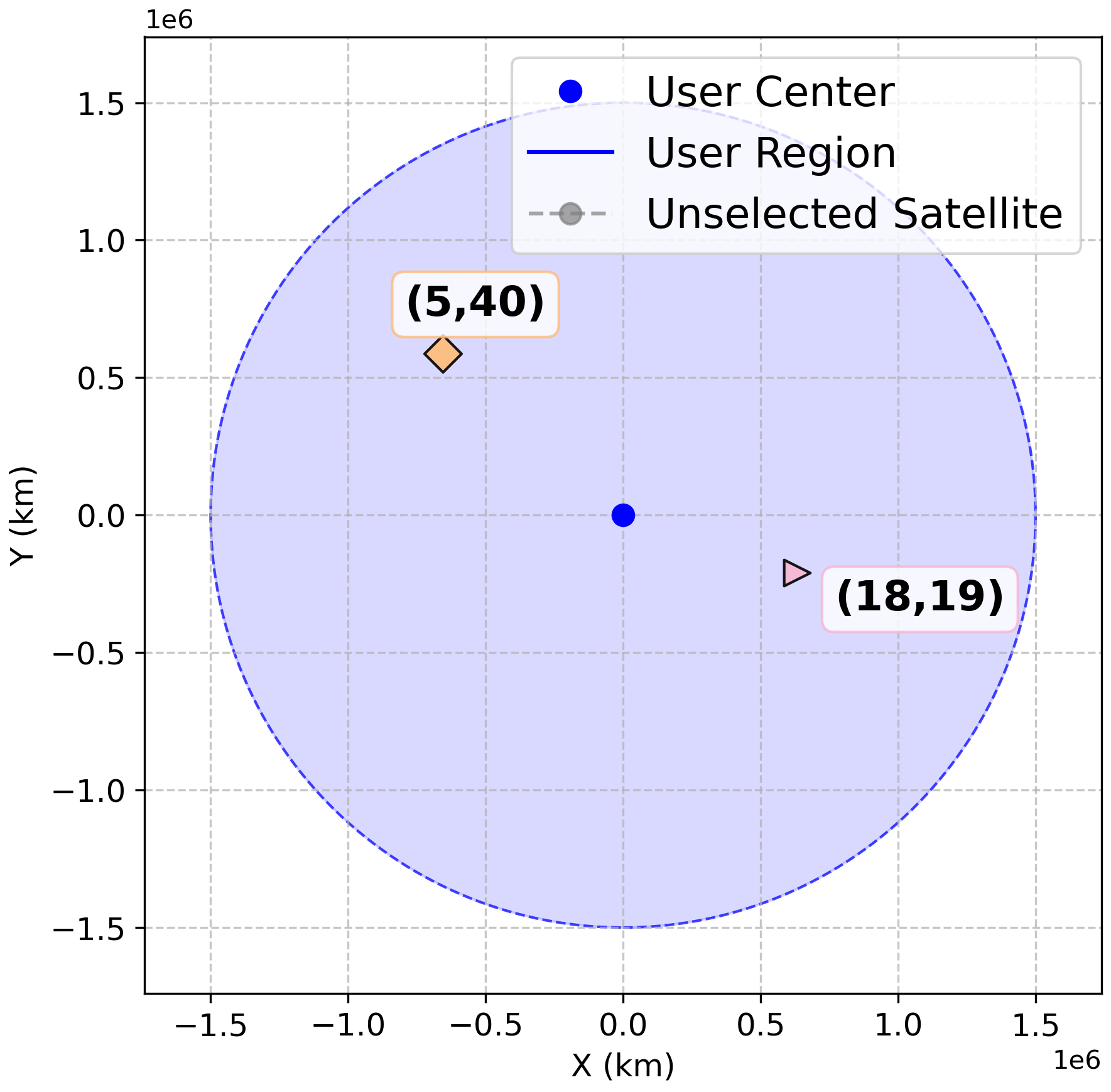}\label{fig:2dt1}} 
		\subfigure[]{\includegraphics[width=0.22\textwidth]{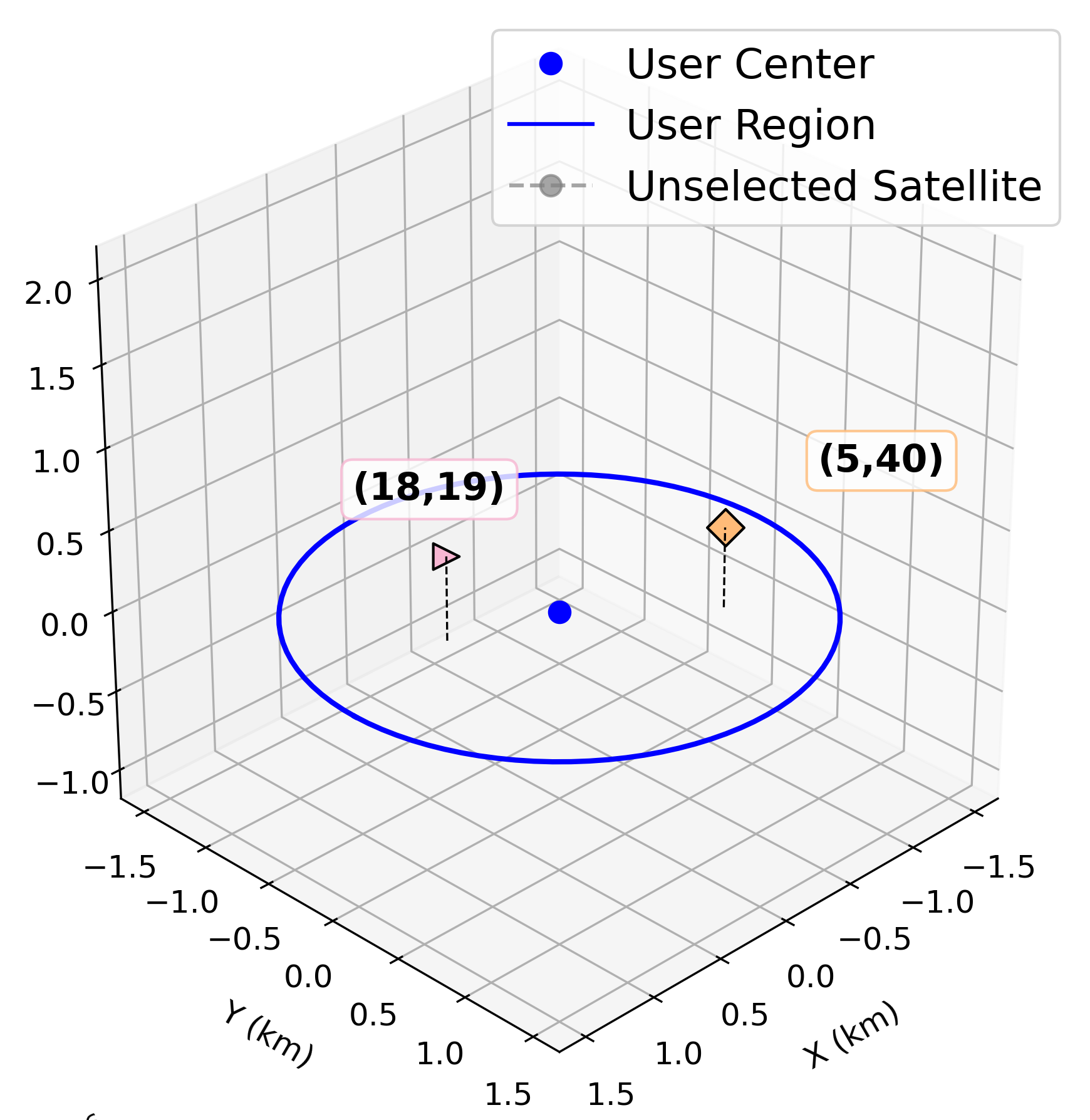}\label{fig:3dt1}} 
		\subfigure[]{\includegraphics[width=0.22\textwidth]{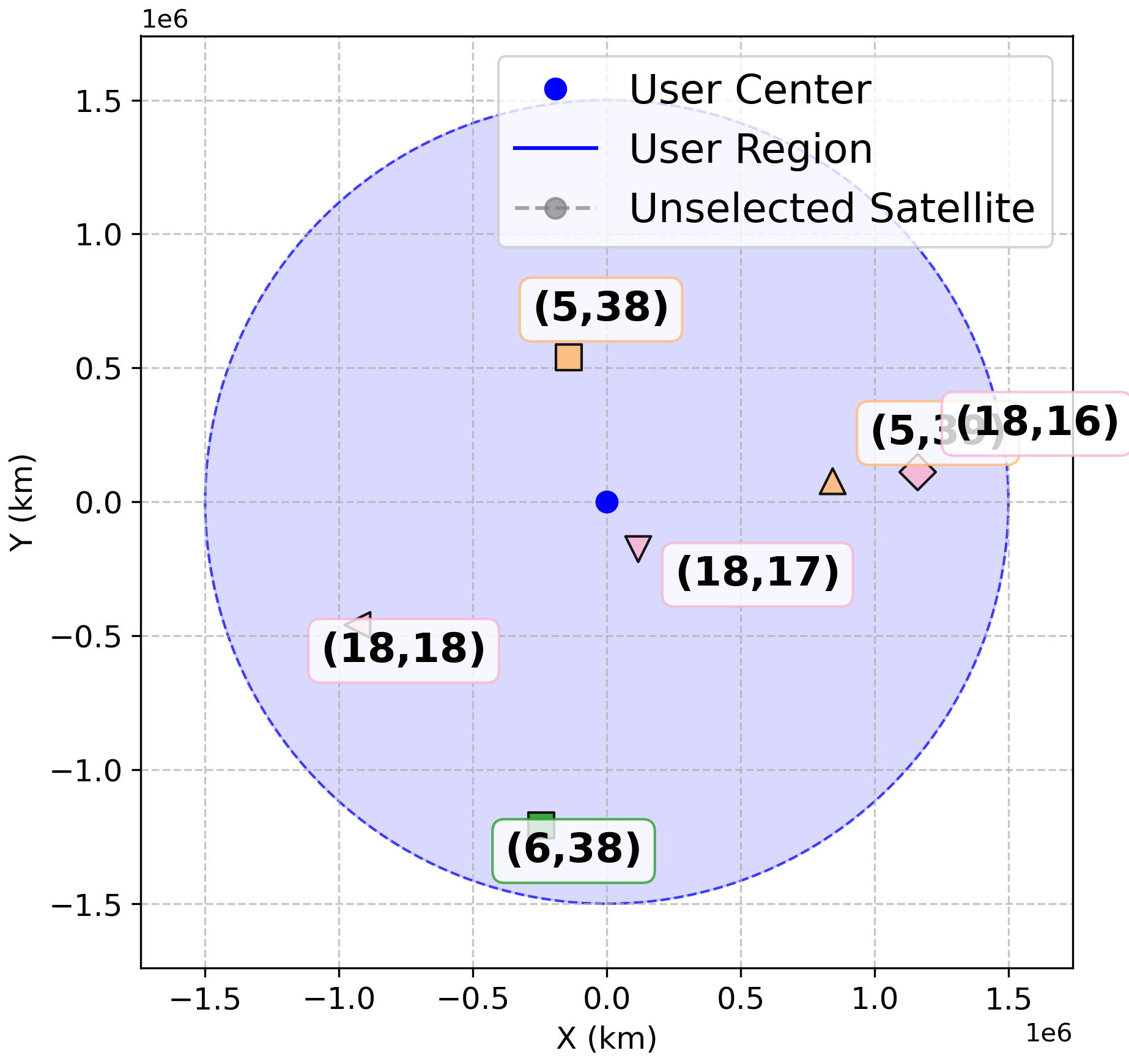}\label{fig:2dt7}} 
		\subfigure[]{\includegraphics[width=0.22\textwidth]{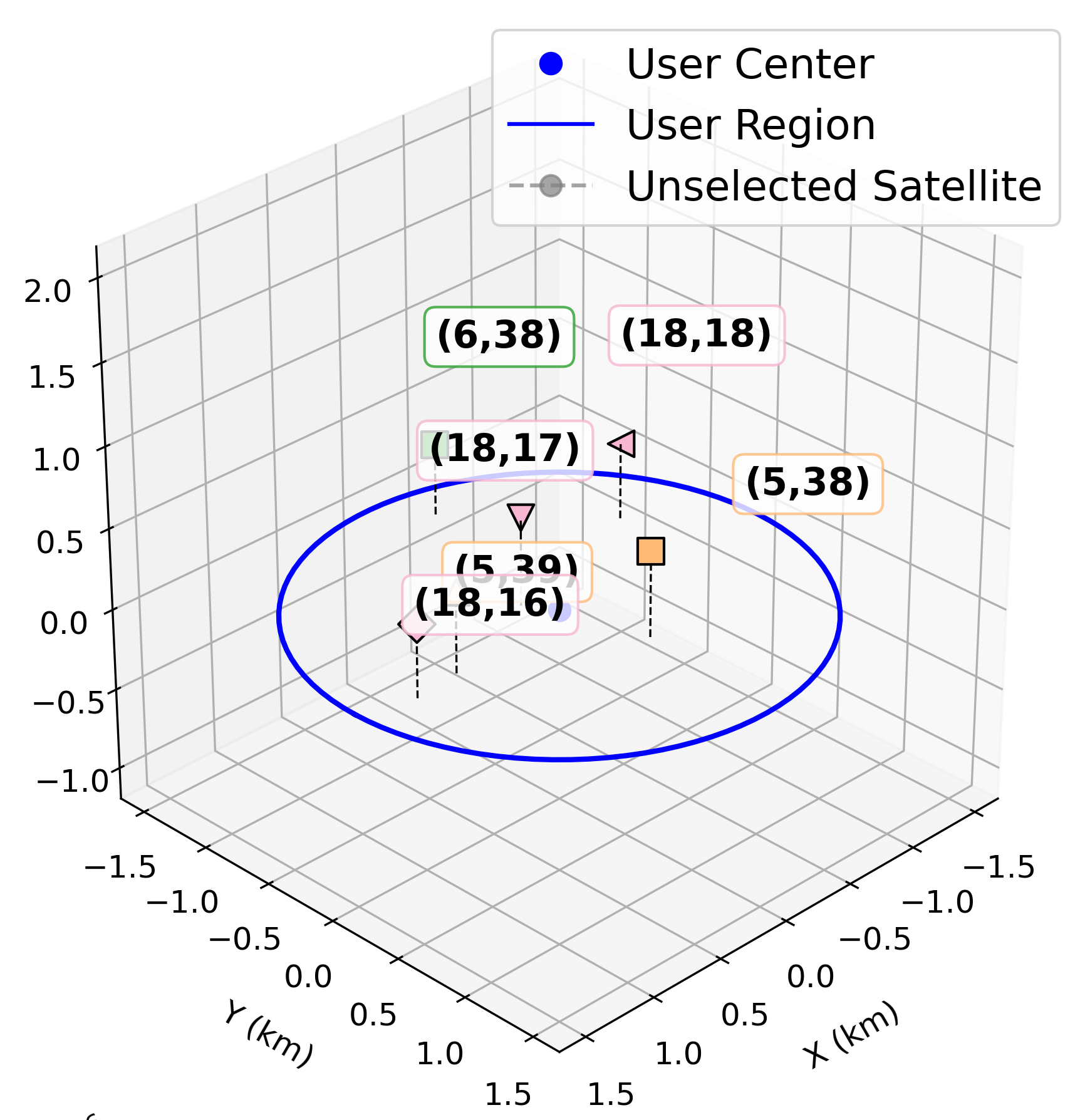}\label{fig:3dt7}} 
		\subfigure[]{\includegraphics[width=0.25\textwidth]{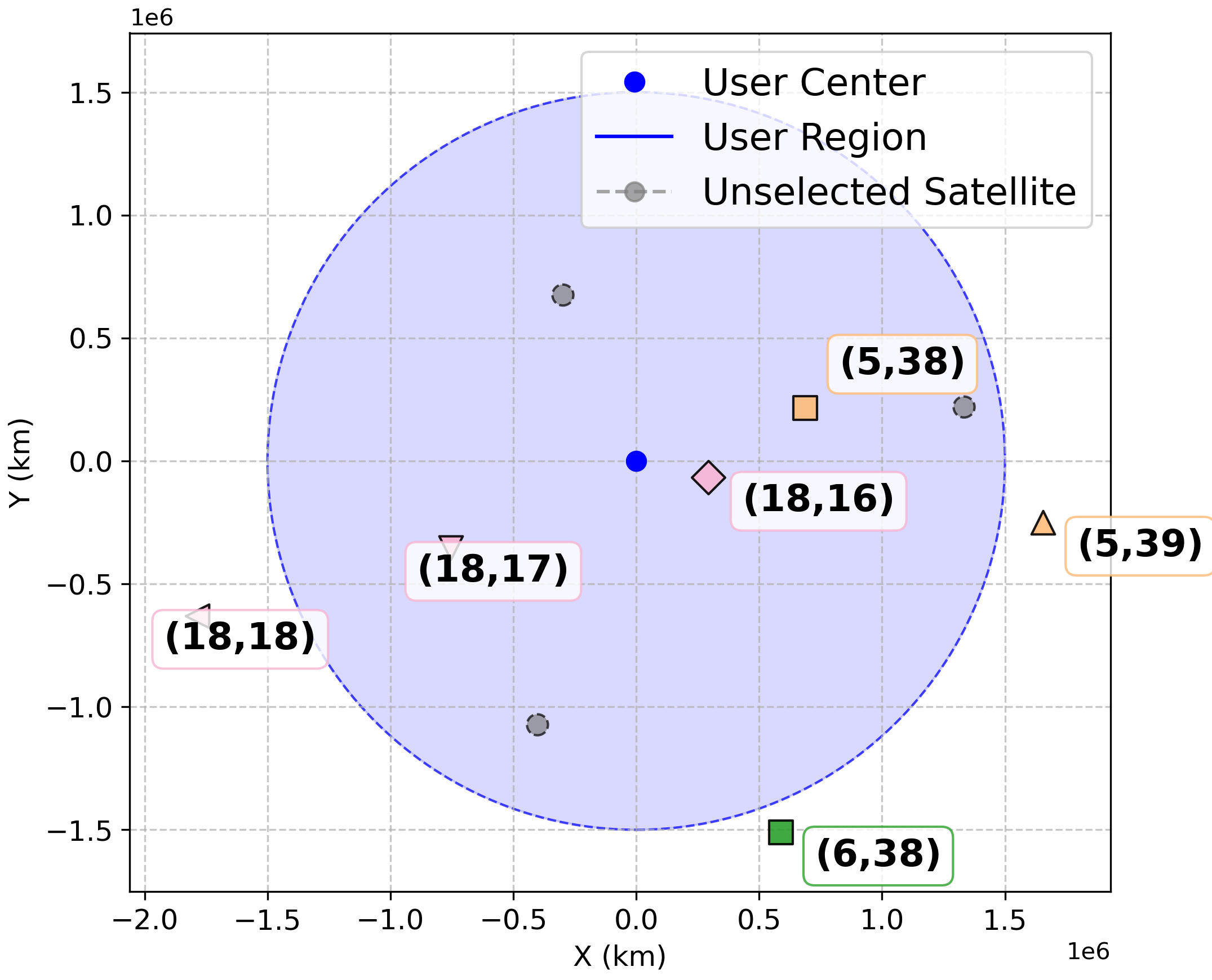}\label{fig:2dt9}}
		\subfigure[]{\includegraphics[width=0.22\textwidth]{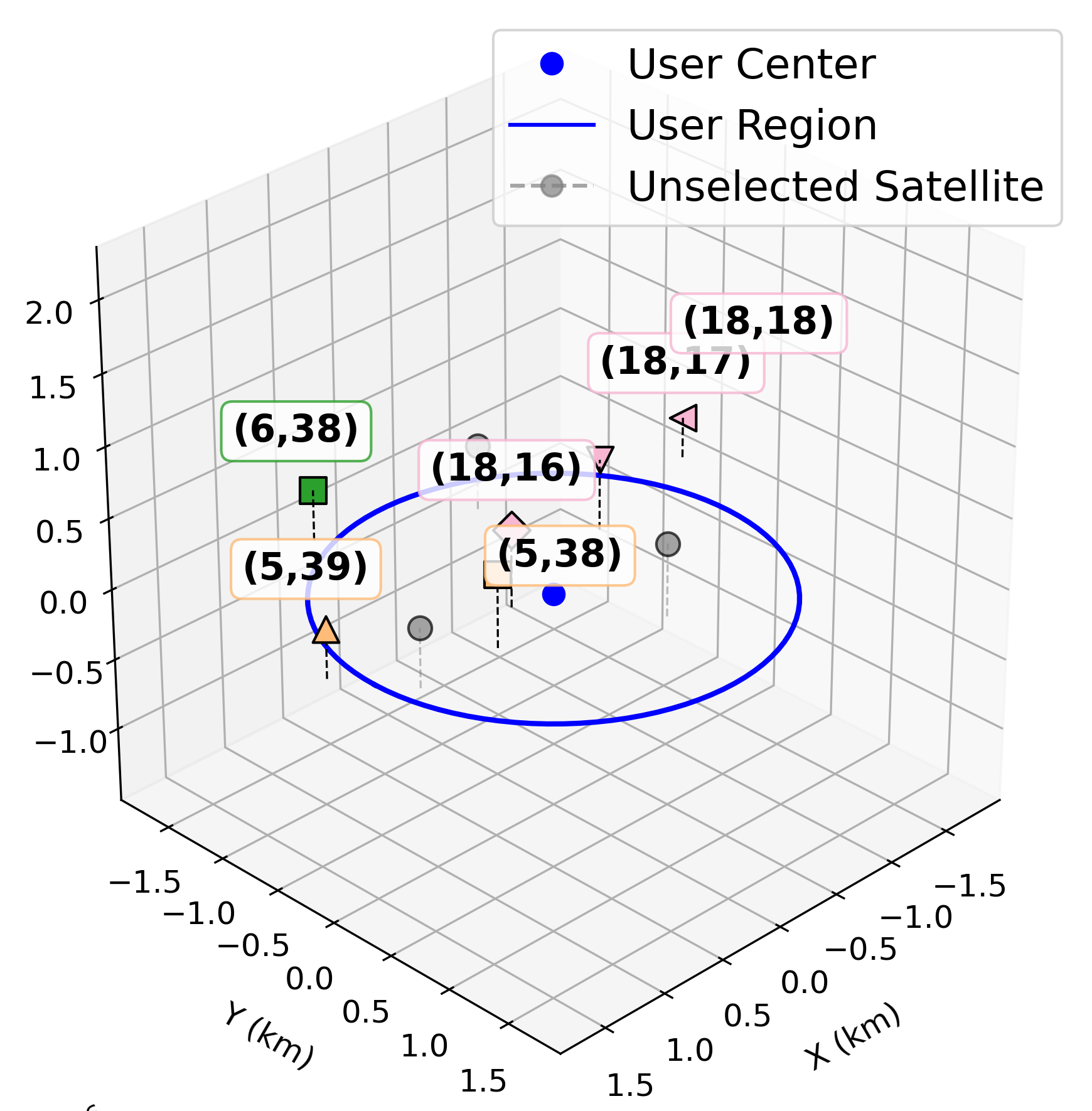}\label{fig:3dt9}} 
		\subfigure[]{\includegraphics[width=0.22\textwidth]{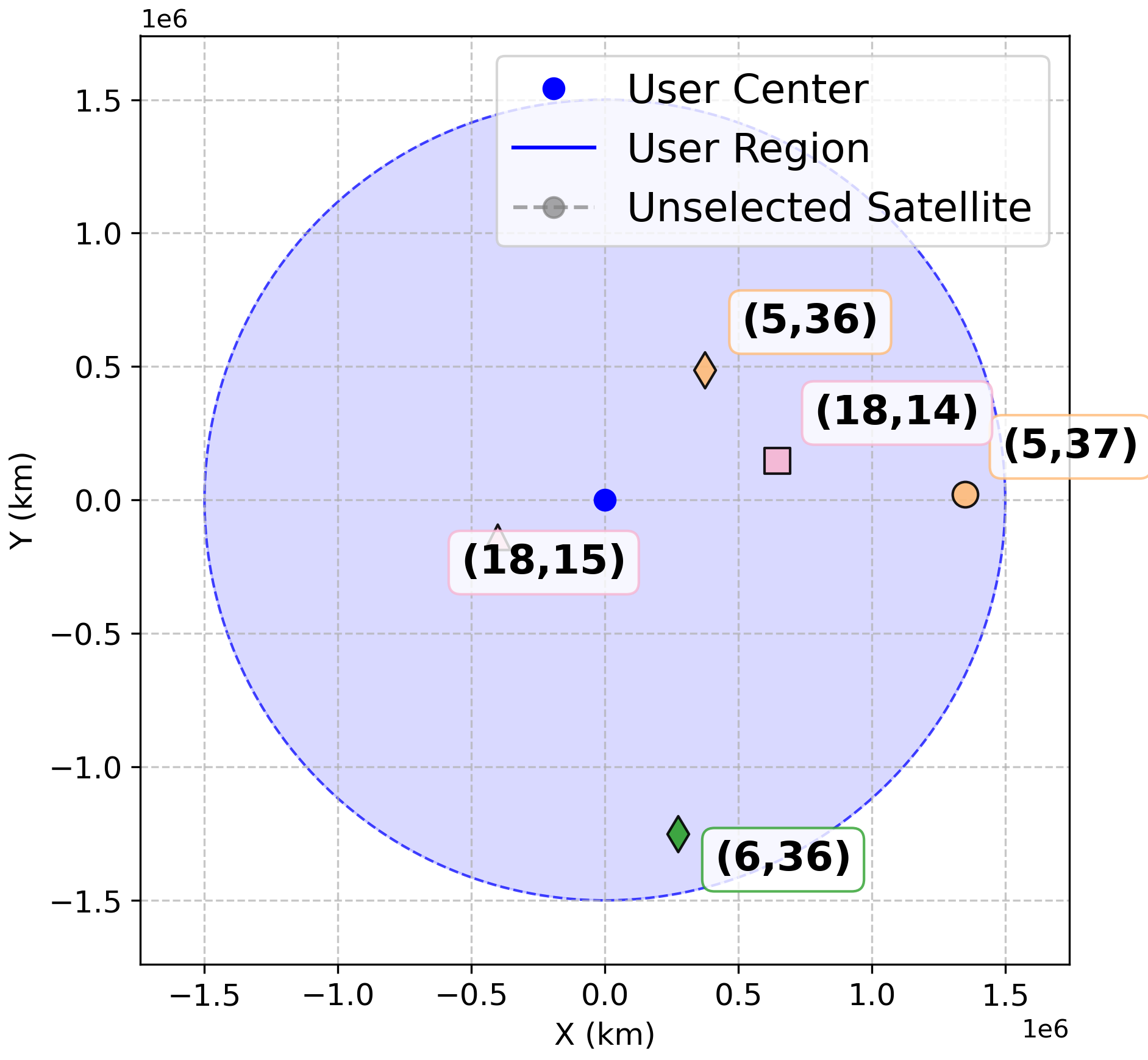}\label{fig:2dt13}} 
		\subfigure[]{\includegraphics[width=0.22\textwidth]{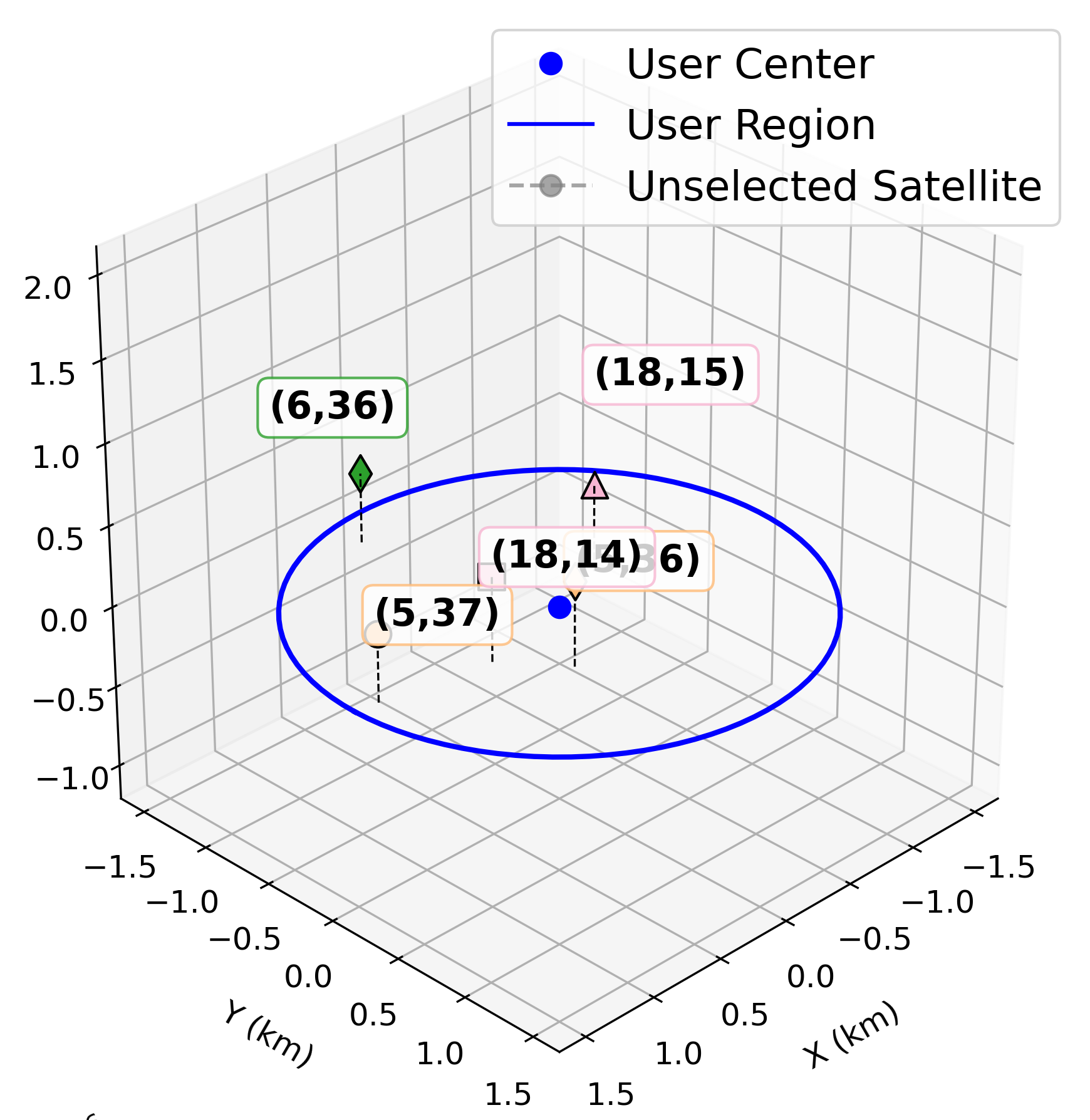}\label{fig:3dt13}} 
\caption{Visualized satellite trajectories: (a) 2D at $t=1$, (b) 3D at $t=1$, (c) 2D at $t=7$, (d) 3D at $t=7$, (e) 2D at $t=9$, (f) 3D at $t=9$, (g) 2D at $t=13$, and (h) 3D at $t=13$. ($\varphi=\frac{\pi}{5}$ rad).}\label{fig:satellite_movement}
\end{figure*}

\subsubsection{Comparative Evaluation of Satellite Selection Strategies under Dynamic Visibility}
To further gain the insight of \methodname{}, we conduct comparative experiments that highlight its capability not only in adaptively selecting satellites, but also in dynamically determining the number of active satellites with visibility variations and constellation size constraints. 
For this purpose, we select the following four schemes to compare with our proposed \methodname{} approach. 
\begin{itemize}
\item $\epsilon$-Markov \cite{eMarkov}: It follows the general framework of the proposed algorithm. The execution of either the E-stage or the C-stage is determined in a stochastic manner. 
\item $s_{\varphi}$-nearest: It selects the $s_{\varphi}$ satellites closest to the UE cluster center. 
\item $2$-nearest: It always selects the two closest satellites to the UE cluster center.
\item $s_{\varphi}$-random: It randomly selects $s_{\varphi}$ satellites without considering their locations.
\end{itemize}
Here, $s_{\varphi}$ denotes the number of service satellites determined by \methodname{} as a function of the cone angle $\varphi$.
For a fair comparison, all methods (except the $2$-nearest scheme) are assumed to select the same number of satellites, $s_{\varphi}$, with an identical \stagename{} process. 
The $2$-nearest method fixes the number of service satellites to $2$ irrespective of the cone angle.

Fig.~\ref{fig:PerformanceBenchmark} presents the sum rates of the proposed \methodname{} framework in comparison with the above baselines  under varying cone angles $\varphi \in \left[\tfrac{\pi}{10}, \tfrac{7\pi}{20}\right]$ for a network with $|\mathcal{U}|=35$. 
From the results, we observe that \methodname{} consistently achieves the highest sum rate across the entire range of cone angles. 
In particular, the performance gap between \methodname{} and the $2$-nearest scheme is remarkable.
Since the $2$-nearest scheme fixes the number of serving satellites to $2$, it fails to exploit the potential spatial diversity provided by larger visibility regions when $\varphi$ increases.
By contrast, \methodname{} dynamically adapts the number of selected satellites, thereby striking a balance between interference mitigation and multi-satellite diversity. 
This adaptivity translates directly into substantial throughput gains. 
Moreover, while the $s_{\varphi}$-nearest and $s_{\varphi}$-random baselines partially reflect the constellation size constraint by fixing the number of selected satellites to $s_{\varphi}$, they ignore the coupled impact of cone-angle-shaped visibility and interference. 
As a result, their performance tends to plateau or even degrade at larger cone angles. 
By contrast, \methodname{} jointly incorporates constellation size and visibility geometry into the selection and allocation process, maintaining consistently superior performance across all cone-angle settings.
In addition to the above observations, Fig.~\ref{fig:PerformanceBenchmark} shows that $\epsilon$-Markov method achieves competitive but inferior results compared to \methodname{} owing to its stochastic decision-making between exploration and consolidation. This stochastic behavior prevents consistent exploitation of favorable satellite constellations. 
These results collectively demonstrate that the proposed \methodname{} framework is not only adaptive in the number of satellites but also better suited to address visibility constraints and constellation dynamics than benchmarks.  

Finally, to evaluate the influence of the inherent \leoname{} dynamism on the number of selected service satellites, we simulated the time-varying selection results obtained by \methodname{}. 
Fig.~\ref{fig:satellite_movement} illustrates the constellation snapshot of the chosen satellites under the proposed workflow, with simulation results recorded from $t=1$ to $t=13$ minutes. 
Two important observations can be made. 
First, the number of selected satellites varies continuously over time. This is a direct consequence of the dynamic evolution of the constellation and the changes of satellites visibility. 
Second, as highlighted in Fig.~\ref{fig:3dt9} at $t=9$ minutes, there exist three visible satellites that are not selected by \methodname{} above the service area. 
According to the $s_{\varphi}$-nearest rule, they should have been chosen instead of satellites $s_{6,38}$ and $s_{18,18}$. 
This indicates that merely selecting the geographically closest satellites does not always guarantee optimal performance, thereby underscoring the advantage of the proposed strategy.

\section{Conclusions}\label{Conclusions}
We investigated the joint optimization of satellite selection, UA, BA, and PA in \leoname{} constellations under dynamic visibility constraints. 
To tackle the inherent NP-hardness of the formulated problem, we proposed the \methodname{} framework that integrates Markov approximation with a matching game-based approach for UA and BA. 
Simulation results on the Walker constellation demonstrated that \methodname{} achieves stable convergence under various channel conditions and significantly outperforms four state-of-the-art baselines. 
In particular, the proposed algorithm provides, on average, a $7.85$\% sum-rate gain over the best-performing benchmark with affordable computational complexity. Moreover, dynamic simulations show that the selected service satellites adapt seamlessly to time-varying constellation visibility, thereby verifying the robustness and generality of the approach.
An additional insight is that the commonly used closest-satellite selection rule does not necessarily yield the best performance, underscoring the necessity of adaptive and visibility-aware satellite selection in practical \leoname{} systems.

\appendices

\bibliographystyle{IEEEtran}
\bibliography{IEEEexample}

\end{document}